\documentclass[12pt]{article}
\usepackage{color}
\usepackage{amsfonts}
\usepackage{amssymb}
\usepackage{amsmath}
\usepackage[mathscr]{eucal}
\usepackage{amsfonts}
\usepackage{amsmath,amsxtra,latexsym,amsthm, amssymb, amscd}
\usepackage[colorlinks,citecolor=blue,filecolor=black,linkcolor=blue,urlcolor=black]{hyperref}
\usepackage{pgfplots} 
\usepackage{url}
\usepackage[round]{natbib}
\usepackage{fancyhdr}
\usepackage{graphicx}
\usepackage{hyperref}
\usepackage{algpseudocode}

\usepackage{caption}
\usepackage{subcaption}
\usepackage{epstopdf}

\usepackage{soul}

\usepackage{titlesec}

\usepackage[T1]{fontenc}
\usepackage[mathscr]{eucal}
\usepackage{amsfonts}
\usepackage{amsmath,amsxtra,latexsym,amsthm, amssymb, amscd}
\usepackage[round]{natbib}
\usepackage{fancyhdr}
\usepackage{graphicx}
\usepackage{lipsum}

\usepackage[latin1]{inputenc}

\usepackage{lmodern}
\usepackage[none]{hyphenat} 
\usepackage{authblk}

\newtheorem{ex}{Example}
\newcommand{\bex}{\begin{ex}}
\newcommand{\eex}{\end{ex}}

\usepackage{fancyhdr}

\newtheorem{definition}{Definition}
\newtheorem{lemma}{Lemma}

\newtheorem{proposition}{Proposition}
\newtheorem{assum}{Assumption}
\newtheorem{corollary}{Corollary}

\newcommand{\rr}{\mathbb{R}}

\newcommand{\ma}{\max\limits}

\begin{document}
\title{Wariness and Poverty Traps
 }
\author{Hai Ha PHAM\footnote{Department of Mathematics, International University, Hochiminh city  and Vietnam National University Hochiminh city, Vietnam. \textit{Email}: phha@hcmiu.edu.vn. Address:  Linh Trung, Thu Duc, Hochiminh, Vietnam.}\quad \quad \quad   Ngoc-Sang PHAM \footnote{Corresponding author. EM Normandie Business School, M\'etis Lab. {\it Emails}: pns.pham@gmail.com, npham@em-normandie.fr. Phone:  +33 2 50 32 04 08. Address:  EM Normandie (campus Caen), 9 Rue Claude Bloch, 14000 Caen, France.}}


\date{\today}
\maketitle
\begin{abstract}
We investigate the effects of wariness (defined as individuals' concern for their minimum utility over time) on poverty traps and equilibrium multiplicity in an overlapping generations (OLG) model. We explore conditions under which (i) wariness amplifies or mitigates the likelihood of poverty traps in the economy and (ii) it gives rise to multiple intertemporal equilibria. Furthermore, we conduct comparative statics to characterize these effects and to examine how the interplay between wariness, productivity, and factor substitutability influences the dynamics of the economy.\\
\newline
{\bf Keywords:} wariness, overlapping generations,  economic growth, multiple equilibria, poverty trap, CES production function, capital intensity, elasticity of factor substitution.\\
\textbf{JEL Classifications}: D14, D5, E71, O41.
\end{abstract}

\section{Introduction}

Poverty trap is one of the fundamental issues in economics, in particular in development economics. The existing literature \citep{Azariadis1996,as05} provides several mechanisms to explain why some countries remain poor, including (1) weak, corrupt, or predatory institutions, (2) weak financial system, (3) poor countries often fail to adopt modern technologies, (4) lack of human capital (poor health and low schooling), (5) presence of high-fixed costs, ...

However, the literature pays little attention to the role of individual's preferences. The standard macroeconomic models \citep{cm02} assume that individuals maximize the discounted sum of their utilities. However, there is evidence, not only in economics but also in other fields, suggesting that individuals may care not only about the discounted sum of utilities but also about the worst outcome experienced over time. For instance, in economics, \cite{gs89} study the maximin expected utility. In psychological science, \citet{Kahnemanetal93}, \cite{RedelmeierKahneman96} show that individuals evaluating past experiences (e.g., medical procedures) put strong emphasis on the worst moment (the "peak") and the final moment, rather than integrating utility smoothly across time. In health contexts, \cite{DolanKahneman2008} report that overall well-being assessments are dominated by experiences of severe pain.


Motivated by this evidence and following \cite{pp24}, we consider an overlapping generations (OLG) model in which individuals take care not only of the discounted sum of utilities, but also of the minimum utility over time and then explore the effects of this behavior on poverty traps. Formally, when a consumer lives for two periods, we assume that her(his) intertemporal utility is given by 
\begin{align}\label{modeling}
\text{Intertemporal utility: }U(c,d)&=(1-\lambda){\big(u(c)+\beta u(d)\big)}+\lambda \min\big\{u(c),u(d)\big\},
\end{align}
where $c, d$ represent the consumer's consumption in the present and the future respectively, $\beta$ is the rate of time preference.\footnote{The modeling of wariness in (\ref{modeling}) can be considered as a finite-horizon version of the utility function $\sum_{t\geq 0}\beta^tu(c_t)+a \inf_{t\geq 0}u(c_t)$ described in Section 3.1.1 in \cite{wariness} or  Example 1 in \cite{wariness4} and \cite{wariness2}, where wariness can be viewed as a kind of ambiguity aversion on a set of discount factors.} The parameter $\lambda\in [0,1]$ can be interpreted as the wariness degree of the household: the higher the value of $\lambda$, the more the household cares about her(his) minimum utility across time. When $\lambda = 0$, we recover the standard case \citep{cm02}. When $\lambda=1$, the agent cares only about the minimum of her consumptions over time.



Our paper aims to explore the role of wariness on the dynamics of capital stocks, poverty traps, and possibility of multiple equilibria. 

The first part of our paper focuses on the convergence and multiplicity of equilibrium. We provide conditions under which there exists a unique equilibrium and the capital path converges. This happens when, for example, the utility function $u$ has the intertemporal elasticity of substitution bounded below by $1$ (i.e.,  $xu'(x)$ is increasing) and the production function $f$ satisfies the condition that the function $x(1+f'(x))$ is increasing. However, these conditions may be violated under CES production functions with the elasticity of factor substitution (EFS hereafter) lower than $1$.\footnote{Empirically, \cite{KlumpIrmen07} estimated, by using data of the U.S. economy from 1953 to 1998, that the elasticity of factor substitution is significantly below unity. See also \cite{Klumpetal12},  \cite{Knoblachetal2020} more complete reviews.}

In a standard OLG model without wariness, it is known (see \cite{cm02}'s Section 1.5.2 among others) that the uniqueness of intertemporal equilibrium is guaranteed if the intertemporal elasticity of substitution is bounded below by $1$. We contribute by arguing that, under the presence of wariness, this insight is no longer true. The intuition is that wariness affects the saving behavior of households and then the dynamics of capital stocks, which creates a room for multiple equilibria. With the same parameters, there may exist two different equilibria: an equilibrium where economic activities collapse (in the sense that the capital path converges to zero) and another equilibrium whose capital path converges to a stable steady-state. 

In the second part of our article, we show how wariness affects poverty traps. Due to the presence of wariness as in (\ref{modeling}), the dynamics of capital is represented by a nonlinear, piecewise-smooth dynamical system. There are several situations depending on the interplay between wariness and capital return.

In the first situation (characterized by low productivity and low wariness), we show that wariness exerts a positive effect on the dynamics of the economy: the higher the degree of wariness, the smaller the set of poverty traps, and the greater the likelihood of escaping them. The intuition is as follows. When productivity is low, the return on capital is also low. Consequently, households expect lower income in old age relative to their income when young. Under low levels of wariness, households place greater weight on consumption in old age. This leads them to reduce current consumption and increase savings when young, thus  boosting investment and mitigating the risk of falling into a poverty trap. Our result has an important policy implication: for a developing country, a moderate wariness would be good for the economic dynamics.

In the second situation (high productivity and low wariness) wariness instead has a negative effect on the dynamics of the economy: the higher the level of wariness, the larger the set of poverty traps. In other situations (intermediate or high productivity and wariness), we also prove that a poverty trap exists. In addition, we show that when the productivity is very low, the economy collapses.

 Although the literature \citep{Azariadis1996,as05, cm02}  discussed the role of the elasticity of factor substitution in  poverty traps,  it did not explicitly provide comparative statics. Our contribution is to fill this gap. Under logarithmic utility and CES (constant-elasticity-of-substitution) production functions, we manage to conduct comparative statics to show how the factor substitutability affects poverty traps. 
\begin{enumerate}
\item We find that the higher the capital intensity (or fraction of goods that have been automated by machines), the larger the set of poverty traps. This leads to an interesting implication: when more goods can be automated, we need more capital to avoid poverty trap. 

\item\label{introductionpoint2} However, the effects of the elasticity of factor substitution (EFS) on the poverty trap can be negative or positive, depending on the interaction between several economic variables (see Lemmas \ref{CesFull}, Propositions \ref{cesll}, \ref{ces_l2}). When the EFS is quite high or the productivity is quite low, we prove that the higher the EFS, the smaller the set of poverty traps, the better chance to prevent a poverty trap. However,  with  intermediate EFS and productivity, a higher elasticity of substitution can enhance the possibility of poverty trap.
\end{enumerate}
We complement our theoretical results by running several numerical simulations, which helps us to better understand the role of wariness and factor substitutability on the dynamics of the economy.



\paragraph{Related literatures.}

To our knowledge, \cite{pp24} is the first introducing the modeling  (\ref{modeling}) in an OLG model and studying its effects on  economic growth.  However,  they did not study poverty traps and excluded CES production functions. We extend \cite{pp24}   by studying the possibility of poverty traps and multiple equilibria. Our analyses, including comparative statics, cover the case of CES production functions.

Our paper is related to the literature on endogenous discounting.  Indeed, we can rewrite (\ref{modeling1}) as follows: 
 \begin{align}
    U(c,d)=  \beta_1(c,d) u(c) +\beta_2(c,d) u(d), 
 \end{align}
 where the endogenous discount factors $\beta_1(c,d), \beta_2(c,d)$ are given by:  If $c\leq d$, then $\beta_1(c,d)=1,\beta_2(c,d)=(1-\lambda)\beta $. If $c>d$, then $\beta_1(c,d)= (1-\lambda),\beta_2(c,d)=(1-\lambda)\beta+\lambda$. It means that the discount factors are functions of consumptions.\footnote{Note that the functions $\beta_1(c,d), \beta_2(c,d)$ are not differentiable.}
 
 \cite{wariness, wariness3, wariness2}, \cite{wariness4} consider the utility function of the form $\sum_{t\geq 0}\beta^tu(c_t)+a \inf_{t\geq 0}u(c_t)$ in general equilibrium models with infinitely-lived agents but without production. However, they focus on the effects of parameter $a$ on asset bubbles.  \cite{tm22} study the optimal capital path of an infinite-horizon model with Ramsey-Rawls criterion with the objective function is $\sum_{t\geq 0}\beta^tu(c_t)+a \inf_{t\geq 0}u(c_t)$. 
\cite{tm22} provide conditions under which the optimal capital stock is constant over time or coincides with the solution to the Ramsey problem (that is, when the parameter $a=0$). Unlike these articles, we use an OLG model and explore the interplay between wariness and poverty traps.

Our paper concerns the literature of the role of endogenous discounting on economic dynamics. \cite{levanJME2011}, \cite{bhh25}, \cite{5authors2025}  focus on optimal growth models with endogenous discount rates (see more references in these papers).\footnote{For example, \cite{levanJME2011} assume that the infinitely-lived agent has the utility function $\sum_{t\geq 1}\beta(x_1)\cdots \beta(x_t)u(c_t)$ while \cite{bhh25} consider $\sum_{t\geq 1}\beta(c_1)\cdots \beta(c_{t-1})u(c_t)$, where $x_t, c_t$ are the capital stock and consumption at date $t$, and $\beta(\cdot)$ is the discount function.}  They show that the optimal path is monotonic over time (see Proposition 8 in \cite{levanJME2011}, Proposition 2.5 in \cite{bhh25}, Theorem 3.1 in \cite{5authors2025}). Under mild conditions, they prove the existence of a poverty trap.  By contrast, we work with an OLG model, which is, in general, more tractable and allows us to obtain more detailed analyses, including comparative statics, equilibrium multiplicity, and especially to explore the role of wariness on poverty traps.

Our paper is related to the role of elasticity of factor substitution and economic growth.\footnote{See \cite{Klumpetal12}, \cite{Knoblachetal2020} for surveys on CES production functions and elasticity of factor substitution.}  Several papers \citep{KlumpPreissler00, KlumpdeLaGrandville00, KlumpSaam08, KlumpIrmen07, Klumpetal08} study the effect of the EFS on the growth rates and the per capita income. For example,  in a neoclassical growth model \`a la Solow, \cite{KlumpdeLaGrandville00} prove that the elasticity of substitution has a positive impact on the capital-labor and income
per head. However, in an OLG model, \cite{Miyagiwa1Papageorgiou03} show that  there exists no such monotonic relationship between factor substitutability and growth. We contribute by showing  the (non-monotonic) effects of the EFS on poverty traps (see point \ref{introductionpoint2} above).

Our article has a link with the literature on the impacts of uncertainty on economic development because uncertainty may generate a wariness for economic agents.  Several empirical papers \citep{kumar23,bloometal22} document a negative relationship between high uncertainty and firms' investment.  From a theoretical point of view, \cite{fukuda08} uses a OLG model and assumes that producers face a Knightian uncertainty in their technologies. Considering a logarithmic utility function, he shows that a poverty trap can arise.  
While these papers focus on the behavior of firms, we study how wariness in households' preferences affects poverty trap. Our novel insight is that when the productivity is low, a low level of wariness in the households'  preferences may be beneficial to the economy and reduce the possibility of poverty trap because it motivates households to save more and, by the way, improve the investment.

The remainder of the paper is organized as follows. In Section \ref{model}, we introduce wariness in an OLG model and provide basic properties of intertemporal equilibrium.  Section \ref{section3} studies the existence, the uniqueness, the multiplicity, and the convergence of equilibrium. Section \ref{section4trap} focuses on the impacts of wariness on poverty traps under general settings and provides numerical simulations. Section \ref{conclu} concludes. Formal proofs are gathered in the appendix section.

\section{An OLG Model with wariness}\label{model}

\subsection{Household and wariness}
 At period $t$, $N_{t}$ individuals are born. We assume that the population growth rate is  constant over time and denote $n\equiv N_{t+1}/N_{t}$. Each consumer-worker lives two
periods. When young, he(she) supplies one unit of labor, earns a labor income, consumes $c_{t}$ and saves $s_{t}$.  When old, he(she) receives the income from her saving and consumes $d_{t+1}$. 

Following \cite{pp24}, we introduce wariness in a standard two-period OLG model \citep{cm02} by assuming that the utility of each household born at date $t$ is given by
\begin{align}\label{modeling1}
(1-\lambda) {\big(u(c_t)+\beta u(d_{t+1})\big)}+\lambda \min\big\{u(c_t),u(d_{t+1})\big\}
\end{align}
 where $\lambda\in [0,1]$ represents the wariness of this individual.  
When $\lambda = 0$, we recover the standard case. When $\lambda=1$, the agent only cares about the minimum of her(his) consumption $\min (u(c_t), u(d_{t+1}))$.

When $\lambda<1$, we denote $\gamma=\frac{\lambda}{1-\lambda}$. Then, $\gamma$ varies between $0$ and $+\infty$. It also represents the wariness of the household. The maximization problem of household born at date $t$ is given by
\begin{align}
(P_{c,t}): \quad &\ma_{(c_t,d_{t+1},s_t)}\Big[U(c_t,d_t)\equiv u(c_t)+\beta u(d_{t+1})+\gamma \min (u(c_t), u(d_{t+1}))\Big]\nonumber\\
\label{bc1}&c_{t}+s_{t}\leq  w_{t}, \quad d_{t+1} \leq R_{t+1}s_{t}, \quad  c_t,d_{t+1},s_t\geq 0,
\end{align}
where $w_t$ is the wage at date $t$ while $R_{t+1}$ represents
the capital return between time $t$ and $t+1$.

 We require standard assumptions as in \cite{cm02}.
\begin{assum}\label{HU}The function $u$ is twice continuously differentiable, strictly increasing, strictly concave and $u'(0)=\infty$.
\end{assum}
Under Assumption \ref{HU}, the function $U(c,d)$ is strictly concave. So, for given $w_t,R_{t+1}>0$, the maximization problem of the household born at date $t$ has a unique solution, and then we can define the saving function. 

\begin{definition}\label{sbeta}
(1) For $w_t,R_{t+1}>0$,  denote $s_t=s(w_t,R_{t+1})$ the optimal saving of the household problem $(P_{c,t})$.

(2) Given $\beta >0$, $w>0, R>0$, we define $s_{\beta}(w, R)$  the unique solution of the following equation $u'\left(w - s\right) = \beta R  u'\left(R s\right).$
\end{definition}


Under Assumptions \ref{HU}, $s_{\beta}(w, R)$ is uniquely well-defined. 
Since $u$ is strictly concave, the function $u'\left(w - s\right)- \beta R  u'\left(R s\right)$ is strictly increasing in $s$. By consequence, $s_{\beta}(w, R)$ is strictly increasing in $\beta$. So, we have the following result.

\begin{lemma}\label{rmk_saving}Let Assumption \ref{HU} be satisfied. Let $w>0,R>0$ be given. If $\beta>\beta'>0$, then $s_{\beta}(w, R)>s_{\beta'}(w, R)$.
\end{lemma}

Denote\textcolor{blue}{
\begin{align}
\beta_1\equiv \beta+\gamma,  \beta_2\equiv \frac{\beta}{1+\gamma}.
\end{align}} As in \cite{pp24}, we have the following result showing the optimal solution of households.
\begin{proposition}\label{proposition1} Let Assumption \ref{HU} be satisfied. Let $w_t,R_{t+1}>0$.
	\begin{enumerate}
		\item[(1)]For $\lambda\in [0,1)$ (or, equivalently, $\gamma \in [0,\infty)$), the optimal saving of the household problem $(P_{c,t})$, denoted by $s_t=s(w_t,R_{t+1})$, is given by 
		\begin{align}\label{sfunction}
			s_t=s(w_t,R_{t+1}) = \begin{cases}
				s_{\beta_1}(w_t, R_{t+1}) & \text{ if } R_{t+1} < \frac{1}{\gamma+\beta}\\
				\frac{w_t}{1+R_{t+1}} &\text{ if }\frac{1}{\gamma+\beta} \leq R_{t+1}\leq \frac{1+\gamma}{\beta}\\
				s_{\beta_2}(w_t, R_{t+1})&\text{ if } R_{t+1} >  \frac{1+\gamma}{\beta}
			\end{cases}.
		\end{align}
		Moreover, we observe that
		\begin{enumerate}
		
		\item  $ R_{t+1} \lesseqgtr \frac{1}{\gamma+\beta} \Leftrightarrow s_{\beta_1}(w_t, R_{t+1}) \lesseqgtr	\frac{w_t}{1+R_{t+1}}$. 
	\item If $\frac{1}{\gamma+\beta} \leq R_{t+1}\leq \frac{1+\gamma}{\beta}$, then $s_{\beta_1}(w_t, R_{t+1})\geq 	\frac{w_t}{1+R_{t+1}} \geq s_{\beta_2}(w_t, R_{t+1})$. 
		\item $R_{t+1}\lesseqgtr  \frac{1+\gamma}{\beta} \Leftrightarrow s_{\beta_2}(w_t, R_{t+1})\lesseqgtr	\frac{w_t}{1+R_{t+1}}$. 
		\end{enumerate}
				\item[(2)]When households only care about the  lowest level of consumption over the life-cycle (i.e. when $\gamma=+\infty$, or equivalently, $\lambda=1$), we have $s_t=\frac{w_t}{1+R_{t+1}}$.
	\end{enumerate}

\end{proposition}

When there is no wariness, the optimal saving is $s_t=s_{\beta}(w_t, R_{t+1})$.  Since $\beta_1\geq \beta\geq \beta_2$, Lemma \ref{rmk_saving} implies that $
s_{\beta_1}(w_t, R_{t+1})\geq s_{\beta}(w_t, R_{t+1})\geq s_{\beta_2}(w_t, R_{t+1}).$ 
It means that the saving of the household under the presence of wariness can be higher or lower than that in the case without wariness.  It depends on the relationship between the interest rate $R_t$ and the wariness level $\gamma$.  

Let us explain the intuition in (\ref{sfunction}). The first regime is when $R_{t+1} < \frac{1}{\gamma+\beta}$. This condition can be rewritten as $R_{t+1} < \frac{1}{\beta}$ (the capital return is low) and $\gamma<\frac{1}{R_{t+1}}-\beta$ (the wariness is low). In this case, the saving is $s_{\beta_1}(w_t, R_{t+1}) $ which is higher than the saving in the case without wariness $s_{\beta}(w_t, R_{t+1})$. Indeed, when the capital return is low, the household's expected income when old $d_{t+1}$ would be low while her income when young $c_t$ would be high. So, in the presence of low wariness, the household cares more about her consumption when old. This implies that the household consumes less and saves more when young.

A similar interpretation applies for the third regime, i.e., when $R_{t+1} >  \frac{1+\gamma}{\beta}$ (the capital return is high $\beta R_{t+1}-1>0$ and the wariness is low $\gamma<\beta R_{t+1}-1$).

The intermediate regime (i.e., when $\frac{1}{\gamma+\beta} \leq R_{t+1}\leq \frac{1+\gamma}{\beta}$ which is equivalent to $\gamma\geq \max(\beta R_{t+1}-1,1-\beta R_{t+1})$) can be interpreted as the {\it high wariness}. In this case, the consumptions when young and old are the same, and the saving equals $\frac{w_t}{1+R_{t+1}}$. 

Observe that if the degree of wariness increases, the intermediate regime enlarges and the difference between the two saving functions $s_{\beta_1}$ and $s_{\beta_2}$ raises.

According to Proposition \ref{proposition1} and \cite{pp24}, we present the following results exploring the effects of the wariness on the optimal saving.  

 \begin{corollary}
 \label{wariness_saving}  Let Assumption \ref{HU} be satisfied.  Given $\omega_t>0$ and $R_{t+1}>0$. We denote $s_t(\gamma)$ the optimal saving of the household with the wariness level $\gamma$. Let  $\gamma_1<\gamma_2$. We have different situations.
 	\begin{enumerate}
\item If $ R_{t+1} <\frac{1}{\gamma_1 +\beta }$, then $s_t{(\gamma_1) }  <s_t{(\gamma_2)}$.
\item If $\frac{1}{\gamma_1 + \beta}\leq 
R_{t+1} \leq \frac{1+\gamma_1}{\beta}$,  then $s_t{(\gamma_1) }  = s_t{(\gamma_2)} =\frac{w_t}{1+R_{t+1}}$
\item If $\frac{1+\gamma_1}{\beta}<R_{t+1}$, then $s_t{(\gamma_1)}>s_t{(\gamma_2)}$.
 	\end{enumerate}
 \end{corollary}

By comparing $\beta R_{t+1}$ with $1$, we obtain the following result showing the monotonicity of the saving function with respect to the wariness level $\gamma$. 
\begin{corollary} \label{coro1} Let Assumption \ref{HU} be satisfied. 
	\begin{enumerate}
		\item [(1)]
If $\beta R_{t+1}<1$ then the optimal saving is increasing in $\gamma$ for $\gamma>0$. 

		\item [(2)]
		 If $\beta R_{t+1}>1$ then the optimal saving is decreasing in $\gamma$ for $\gamma>0$. 
	\end{enumerate}
\end{corollary}

\subsection{Production}Technology is represented by a constant returns to scale, concave production function 
$F\left( K,L\right) $ where $K$ and $L$ are the aggregate capital
and the labor forces.  
Given the capital return $R_t$ and the wage rate $L_t$, the representative firm maximizes its profit by choosing the allocation $(K_t,L_t)$. The firm's profit maximization problem is 
\begin{align}  (P_{f,t}): \quad &\ma_{K_t,L_t\geq 0}\Big(F(K_t,L_t)-R_tK_t-w_tL_t\Big)
\end{align}
Denote $k_{t}\equiv K_{t}/L_{t}$ denotes the capital intensity and  \textcolor{blue}{$f(k) \equiv F( k,1)$.}

\begin{assum}\label{assum2}The function $f$ is  twice continuously differentiable, strictly increasing, strictly concave, and $f(k)>0,\forall k>0$.
\end{assum}
Here, we allow the cases where $f'(0)<\infty$ or/and $f'(\infty)>0$, including the CES production function.\footnote{\cite{pp24} assume that $f(0)=0$, $f(\infty)=\infty$, $f'(0)=\infty$ and $f'(\infty)=0$, which rule out the CES production function.}


\subsection{Intertemporal equilibrium}

\begin{definition}\label{def1} An intertemporal equilibrium is a positive sequence $(R_{t},w_{t},c_{t},d_{t+1}, s_{t},K_{t+1},L_{t})
_{t\geq 0}$ which satisfies the following conditions: (1) given the sequence $(R_t, w_t)_{t\geq 0}$, the allocation $(K_t,L_t)$ is a solution to the problem $(P_{f,t})$ and the allocation $(c_t,s_t,d_{t+1}$) is a solution to the problem $(P_t)$; (2) market clearing conditions:%
\begin{align*}
\text{physical capital}& :K_{t+1}=N_{t}s_{t} \\
\text{labor}& :L_{t}=N_{t} \\
\text{consumption good}& :s_{t}+c_{t}+d_{t}/n=f\left( k_{t}\right),
\end{align*}%
\end{definition}

In equilibrium, we have $k_t>0, \forall t$. So, the profit maximization implies that 
\begin{equation}
R_{t}=f^{\prime}(k_{t}) \text{ and }w_{t}=\omega (k_{t}).  \label{73}
\end{equation}%
where the wage function $\omega: \rr_+\to \rr$ is defined by \textcolor{blue}{$\omega(k) \equiv f(k) -kf^{\prime }(k) $ $\forall k.$}

\begin{lemma}Let $k_0>0$ be given and  Assumptions \ref{HU} and \ref{assum2} hold. A positive sequence $(R_{t},w_{t},c_{t},d_{t+1}, s_{t},K_{t+1},L_{t})
_{t\geq 0}$ is an intertemporal equilibrium if and only if 
\begin{subequations}
\begin{align*}
w_t&=\omega(k_t), R_{t}=f'(k_t),\quad  L_t=N_t, \quad K_{t+1}=N_{t}s_{t}\\
c_{t}&=\omega_{t}-s_{t}, \quad d_{t+1} = R_{t+1}s_{t}\\
s_t&=nk_{t+1}=s\big(\omega(k_t),f'(k_{t+1})\big).
\end{align*}
\end{subequations}
\end{lemma}
Thanks to this result, we can redefine the intertemporal equilibrium as a dynamical system.

\begin{definition}\label{def3}A positive sequence of $k_t$ is an  intertemporal equilibrium with perfect foresight  (or equilibrium for short)  if $nk_{t+1}=s\big(\omega(k_t),f'(k_{t+1})\big),\forall t$ where $k_0>0$ is given.
\end{definition}

As in the standard literature, we have the existence result.
\begin{lemma}[existence of intertemporal equilibrium]
Under Assumptions \ref{HU}, \ref{assum2}, there exists an intertemporal equilibrium.
\end{lemma}
The literature of equilibrium existence is large. The standard approach makes use of the fixed point theorems.\footnote{See, for instance, \cite{bs80}, \cite{wilson81}, \cite{BonnisseauRakotonindrainy2017} for OLG models and \cite{bblvs15}, \cite{lvp16} and references therein for infinite-horizon models. The basic idea is to prove the existence of equilibrium in each $T-$truncated economy, and then let $T$ tend to infinity to get an equilibrium.} However, in our framework, there exists an elementary proof which can be found in Proposition 1.2 in \cite{cm02}.

\section{Convergence and multiplicity of equilibrium}\label{section3}
According to Proposition \ref{proposition1}, the equilibrium system $nk_{t+1}=s\big(\omega(k_t),f'(k_{t+1})\big),\forall t$ becomes
\begin{align}\label{kt+1}
nk_{t+1}=s(\omega(k_t), f'(k_{t+1}))= \begin{cases}
			s_{\beta_1}(\omega(k_t),  f'(k_{t+1})) & \text{ if }  f'(k_{t+1}) < \frac{1}{\gamma+\beta}\\
			\frac{\omega(k_t)}{1+ f'(k_{t+1})} &\text{ if }\frac{1}{\gamma+\beta} \leq  f'(k_{t+1})\leq \frac{1+\gamma}{\beta}\\
			s_{\beta_2}(\omega(k_t),  f'(k_{t+1}))&\text{ if }  f'(k_{t+1}) >  \frac{1+\gamma}{\beta}
		\end{cases}
\end{align}
where $k_0>0$ is exogenously given.  

For analytical clarity, we define the concept of regimes to distinguish between the three possible cases in each period.
\begin{definition}\label{regimedefinition}
We say that $(k_t,k_{t+1})$ is in 
\begin{enumerate}
    \item \label{regimedefinition1}the regime 1 if  $f'(k_{t+1}) < \frac{1}{\gamma+\beta}$,
    \item \label{regimedefinition2}the regime 2 if $f'(k_{t+1}) >  \frac{1+\gamma}{\beta}$
    \item \label{regimedefinition3} the regime 3 if $\frac{1}{\gamma+\beta} \leq  f'(k_{t+1})\leq \frac{1+\gamma}{\beta}$.
\end{enumerate}
\end{definition}
Recall of notation: \textcolor{blue}{$
\beta_1\equiv \beta+\gamma,  \beta_2\equiv \frac{\beta}{1+\gamma}.$} The dynamics of capital depends the interplay between the capital return and the thresholds $\frac{1}{\beta_1},\frac{1}{\beta_2}$ where $\frac{1}{\beta_1}<\frac{1}{\beta_2}$. 

The equilibrium system (\ref{kt+1}) leads to a direct consequence.
\begin{corollary}\label{basic1}Let Assumptions \ref{HU}, \ref{assum2} be satisfied.
\begin{enumerate}
\item 
If $\gamma=0$ (no wariness), then $ nk_{t+1} = s_{\beta}(\omega(k_t), f'(k_{t+1})),\forall t$.

\item 
If $f'(0)<\frac{1}{\beta_1}$, then $ nk_{t+1} = s_{\beta_1}(\omega(k_t), f'(k_{t+1})),\forall t.$
\item 
If $f'(\infty)>\frac{1}{\beta_2}$, then $ nk_{t+1} = s_{\beta_2}(\omega(k_t), f'(k_{t+1})),\forall t$.
\item 
If  $\frac{1}{\beta_1} \equiv \frac{1}{\beta+\gamma} \leq f'(\infty  )<  f'(0)\leq \frac{1}{\beta_2}\equiv \frac{1+\gamma}{\beta}$, then the rate of return $f'(k) \in \left[ \frac{1}{\beta_1}, \frac{1}{\beta_2}\right]$ for all $k$. So, we have $ 	nk_{t+1} = \frac{\omega(k_t)}{1+ f'(k_{t+1})}$ $\forall t.$
\end{enumerate}
\end{corollary}

{In some particular cases,  we can explicitly compute the household's saving and obtain a more explicit dynamics of $(k_t)$. Corollary 3 in  \cite{pp24} explores the  dynamics of capital path for Cobb-Douglass production function and logarithm or CRRA utility functions.  Here, we provide the explicit dynamics of capital path for logarithm utility and CES production technology. It is convenient to introduce useful notations. For $b>0$, the function $g_{b}: \rr_+\to\rr_+$ is defined by 
\textcolor{blue}{
\begin{align}\label{func_g}
	g_{b}(x)= \frac{b A(1-a)(ax^{\rho}+1-a)^{\frac{1}{\rho}-1}}{n(1+b)} \text{ }\forall x\geq 0.
\end{align}  }

	\begin{corollary}\label{corollary-ces}
	Assume $u(c)=\ln(c)$ and a CES production function: 
\begin{align}	\label{Ces}
F(K,L)=A\big(aK^{\rho}+(1-a)L^{\rho}\big)^{\frac{1}{\rho}}, \text{ where $A>0$, $a\in (0,1)$ and $\rho\not=0,\rho<1$.}\end{align}
  Recall that the elasticity of factor substitution is $\frac{1}{1-\rho}$ and parameter $a$ represents the capital intensity. In Appendix \ref{specialcase}, we present detailed properties of the function $f(k) \equiv F(k,1)$. The dynamics of equilibrium capital path becomes
 		\begin{align}\label{CES-dynamics}
 		k_{t+1}= \begin{cases}
 			g_{\beta_1}(k_t)& \text{ if } f'(k_{t+1}) < \frac{1}{\gamma+\beta}\equiv \frac{1}{\beta_1}\\
 			\frac{A(1-a)(ak_t^{\rho}+1-a)^{\frac{1}{\rho}-1}}{n (1+f'(k_{t+1})) } &\text{ if }\frac{1}{\gamma+\beta} \leq f'(k_{t+1})\leq \frac{1+\gamma}{\beta}\\
 			g_{\beta_2}(k_t)&\text{ if } f'(k_{t+1}) >  \frac{1+\gamma}{\beta} \equiv \frac{1}{\beta_2}
 		\end{cases}.
 	\end{align}

 \end{corollary}

	\subsection{Uniqueness and convergence of equilibrium}

We provide explicit conditions to ensure the equilibrium uniqueness and the convergence of capital path. Following \cite{pp24}, we obtain the following result.
	\begin{proposition}\label{unique-convergence}
		Assume that
		$f^{\prime}(k)u^{\prime}\big(nkf^{\prime}(k)\big)$ is strictly decreasing and $h(k)\equiv k+kf'(k)$ is strictly increasing in $k$ for any $k>0$. There exists a unique equilibrium and the dynamics of the equilibrium capital path is given by (\ref{kt+1}).  Moreover, $k_{t+1}$ determined by (\ref{kt+1}) is a strictly increasing, continuous function of $k_t$, {\color{blue} denoted by $G(k_t)$}. By consequence, the capital path $(k_t)$ converges.
	\end{proposition}
\begin{proof}See Appendix \ref{section3-proofs}.\end{proof}

Naturally, one might ask whether the assumptions in Proposition \ref{unique-convergence} is well justified. The following result provides an answer.
		\begin{lemma}\label{h_increase}
		
\begin{enumerate}
\item $f^{\prime}(k)u^{\prime}\big(nkf^{\prime}(k)\big)$ is strictly decreasing in $k$ if one of the two following conditions holds: (1) the function $cu'(c)$ is increasing on $[0,\infty)$, (2) the function $kf'(k)$ is increasing on $[0,\infty)$.
\item \label{hincreasing}Assume a CES production function in (\ref{Ces}). The function
			$h(k)= k + kf'(k)$ is increasing on $(0,\infty ) $ if one of the following conditions holds
			\begin{enumerate}
				\item $0<\rho<1$ 
				\item $ \rho<0$ and $1- Aa^{\frac{1}{\rho}} (-\rho)^{-\frac{1}{\rho}+3}(1-2\rho)^{\frac{1}{\rho}-2}\geq 0$.
			\end{enumerate}
\end{enumerate}
		\end{lemma}
\begin{proof}See Appendix \ref{section3-proofs}.\end{proof}
			
Note that Proposition 1.3 in \cite{cm02} provides Assumption H3 to obtain the uniqueness of intertemporal equilibrium.\footnote{Assumption H3 in \cite{cm02}: for all $w>0, k>0$, if $\Delta(k,w)=0$, then $\Delta_k'(k,w)>0$, where $\Delta(k,w)\equiv nk-s(w,f'(k))$ and $\Delta'_k(k,w)\equiv n-\frac{\partial s(w,f'(k)}{f'(k)}f^{\prime\prime}(k)$.} Although their result is general, their assumption H3 is quite implicit and not easy to be verified. Our assumptions in Proposition \ref{unique-convergence} are more explicit than H3 in \cite{cm02}.	

We are now interested in the identification of steady state. For convenience, we introduce some notations.
\textcolor{blue}{
\begin{align}
\label{notationMi} {M}_i &\equiv 	\frac{\beta_i}{1+\beta_i}	\sup_{k>0}\frac{\omega(k)}{k}, \text{ for } i=1, 2, & {M}_3& \equiv \sup _{k>0}\frac{\omega(k)}{k(1+f'(k))}.
\end{align}
}
where recall that $\omega(k)\equiv f(k) -kf^{\prime }(k)$.
\begin{corollary}[Steady state]\label{prop_convergence}
	Under assumptions in Proposition \ref{unique-convergence}, the equilibrium capital path $(k_t)_{t\geq 0}$ converges monotonically to a steady state $k^*$. 
	We have
	\begin{enumerate}
		\item If $f'(k^*)<\frac{1}{\beta_1}$ then $nk^* =	s_{\beta_1}(\omega(k^*), f'(k^*))$. In addition if $k^*>0$ then $ M_1\geq n$.
		\item $f'(k^*)\geq \frac{1}{\beta_2}$ then $nk^* = 	s_{\beta_2}(\omega(k^*)$. In addition if $k^*>0$ then $ M_2\geq n$
		\item If $ \frac{1}{\beta_1}<f'(k^*)< \frac{1}{\beta_2}$ then	$nk^* = \frac{\omega(k_t)}{1+ f'(k^*)} $. In addition if $k^*>0$ then $ M_3\geq n$
	\end{enumerate}

\end{corollary}
\begin{proof}See Appendix \ref{section3-proofs}.\end{proof}

\subsection{Wariness and multiple equilibria} 
			\label{MultipleEquilibrium}
Proposition \ref{unique-convergence} and 
 Lemma \ref{h_increase} suggest that there may be a room for multiple equilibria. In this section, we will address this issue by focusing on the role of wariness which is the key element of our paper.

In the absence of wariness, it is well known that when $cu'(c)$ is increasing (i.e., the inter-temporal elasticity of
substitution $-\frac{u'(c)}{cu^{\prime\prime}(c)}$ is greater or equal to 1),  the uniqueness of intertemporal holds (see, for instance, Proposition 1.3 and Assumption 4 in \cite{cm02}). 

We will argue that a high level of wariness may lead to multiple equilibria, whatever the level of the inter-temporal elasticity of
substitution. Formally, we have the following result.

\begin{proposition}\label{wariness-multipleequilibria}
Assume that $\frac{1}{\beta_1} \equiv \frac{1}{\beta+\gamma} \leq f'(\infty  )<  f'(0)\leq \frac{1}{\beta_2}\equiv \frac{1+\gamma}{\beta}$. 
The positive sequence $(k_t)$ is an equilibrium if and only if 
\begin{align}\label{sameconsumption}
nk_{t+1}=\frac{\omega(k_t)}{1+f'(k_{t+1})}, \forall t\geq 0, \text{ with $k_0>0$ is given.}
\end{align}
By consequence, if the equation $\omega(k_0)=nk(1+f'(k))$ has at least two strictly positive solutions, then there exists at least two equilibria.
\end{proposition}
\begin{proof}See Appendix \ref{section3-proofs}.\end{proof}
This result shows the possibility of multiple equilibria whatever the form of the utility function $u$. A key is that wariness is very high, which implies that the consumptions when young and old are the same, leading to the dynamics (\ref{sameconsumption}). This means the role of wariness on the equilibrium multiplicity is robust. In Section \ref{highwariness}, we will investigate the role of wariness in poverty traps.
	  
The following example shows a possibility of multiple equilibria under wariness. 
\begin{ex}\label{example1}
{\normalfont Consider $u(c)=ln(c)$ and the CES production function as in Corollary \ref{corollary-ces}: $F(K,L)=A((aK^{\rho}+(1-a)L^{\rho})^{\frac{1}{\rho}}, \text{ where $A>0$, $a\in (0,1)$ and $\rho\not=0,\rho<1$}$.

Assume that households only care about the lowest level of consumption over the life-cycle (i.e., when $\gamma=+\infty$, or equivalently, $\lambda=1$).

The positive sequence $(k_t)$ is an equilibrium if and only if 
\begin{align}\label{ces1}
nk_{t+1}=\frac{A(1-a)(ak_t^{\rho}+1-a)^{\frac{1}{\rho}-1}}{1+Aak_{t+1}^{\rho-1}(ak_{t+1}^{\rho}+1-a)^{\frac{1}{\rho}-1}}, \forall t\geq 0.
\end{align}
There may be multiple intertemporal equilibria as shown by the following graphic. Indeed, we drawn this graphic with $n=1.1, A=3, a=0.3$, $\rho=-3$. 
We see in  Figure \ref{fig0} that when $k_0=1$, there are three positive values of $k_1$ satisfying the dynamics of equilibrium capital path, i.e., $nk_{1}=\frac{A(1-a)(ak_0^{\rho}+1-a)^{\frac{1}{\rho}-1}}{1+Aak_{1}^{\rho-1}(ak_{1}^{\rho}+1-a)^{\frac{1}{\rho}-1}}$, which are approximately $0.4,  0.9, 1.6$. So, there exist at least 3 intertemporal equilibria. Among these three equilibria, if $k_1=0.4$, the economy will collapse (i.e., $\lim_{t\to\infty}k_t=0$. In contrast, if $k_1=1.6$, the economy will converge to a stable steady state.
\begin{figure}[h]
	\centering
	\includegraphics[width=0.4\linewidth]{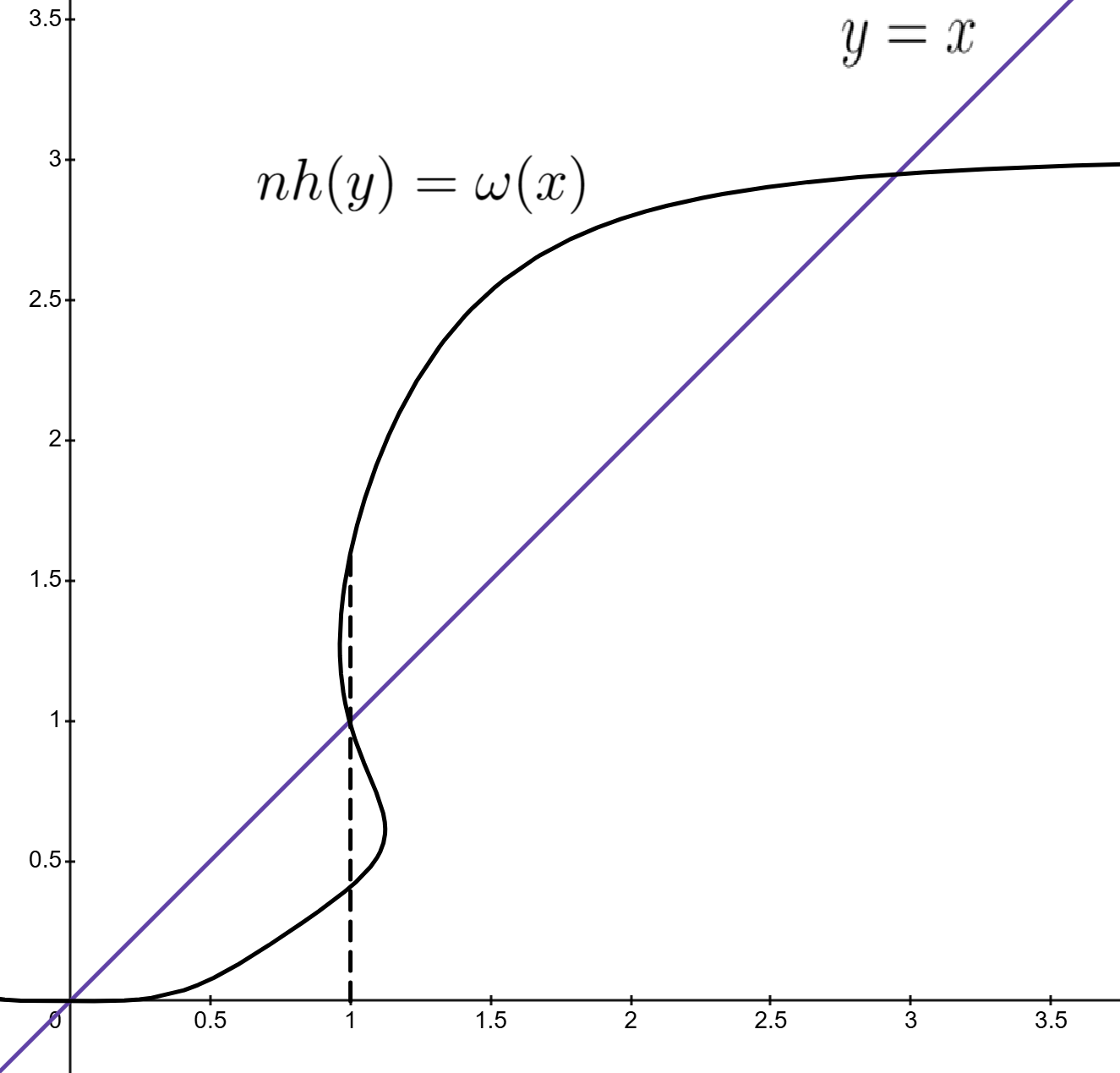}
	\caption{Multiple equilibria}
	\label{fig0}
\end{figure}
}
\end{ex}


It should be noticed that the wariness may lead to the occurrence of multiple equilibria which are in different regimes. The following result shows this possibility.
\begin{proposition}\label{multiple-version2}
Given $k_0>0$. There exist at least two equilibria  which belong to different regimes if at least two of the following situations happen:  \begin{enumerate}
	\item $nk = s_{\beta_1}(\omega(k_0), f'(k))$ has a solution $k^{(1)}_{1}$ which satisfies $f'(k^{(1)}_{1})<\frac{1}{\beta_1}$.
	\item  $nk = s_{\beta_1}(\omega(k_0), f'(k))$ also has a solution $k^{(2)}_{1}$ which satisfies $f'(k^{(2)}_{1})>\frac{1}{\beta_2}$ .
	\item  $nk = \frac{\omega(k_0)}{1+f'(k)}$ has a solution $k^{(3)}_{1}$ which satisfies $\frac{1}{\beta_1}
	 \leq f'(k^{(3)}_{1})\leq \frac{1}{\beta_2}$.
\end{enumerate}
Moreover, if these three conditions hold, there exist at least three equilibria.
\end{proposition}
\begin{proof}See Appendix \ref{section3-proofs}.\end{proof}

\begin{ex}    {\normalfont 
	As in Example \ref{example1}, we consider the CES production function and utility. Let $A=3.4, a=0.4$, $\rho=-3$, $n=1.32, \beta=0.7, \gamma=0.255$ and $k_0 = 1.5$.

    There are three values $k_1$ satisfying the system (\ref{CES-dynamics}). The first value is equal  $k_{1}^{(1)}=1.17$ with $f'(k_{1}^{(1)})-\frac{1}{\beta_1}=-0.15<0$.  The second value is $k_{1}^{(2)}=0.86$ with $f'(k_{1}^{(2)})-\frac{1}{\beta_2}=0.1>0$. The third value is $k_{1}^{(3)}=0.955$, which satisfies $\frac{1}{\beta_1}<f'(k_{1}^{(3)})<\frac{1}{\beta_2}$.

It means that we have at least three equilibria.
    
	
}    \end{ex}

\section{Wariness and poverty trap}
\label{section4trap}

In this section, we present general results showing the role of wariness on poverty trap. It is useful to introduce some notions of growth and collapse.
\begin{definition}[collapse and poverty trap] \label{def2}
\mbox{}
\begin{enumerate}
\item
A value $\bar{k}$ is called a {\it poverty trap} if, for any initial capital stock $k_0<\bar{k}$, we have $k_t<\bar{k}$ for any $t$ high enough.
\item
The economy collapses if $\lim\limits_{t\rightarrow \infty}k_t=0$. 

\end{enumerate}
\end{definition}
 Our formal definition of trap means that a poor country ($k_0\leq \bar{k}$) continues to be poor. It is in line with the notion of poverty trap  in \cite{as05}: A poverty trap is a self-reinforcing mechanism that causes poverty to persist. 
 
 In what follows, our aim is to identify the conditions under which this critical threshold $\bar{k}$ exists and how this threshold depends on wariness and factor substitutability. This allows us to understand how to prevent poverty traps. 

To address the issue of poverty trap, let us start by looking at the benchmark case  where this is no wariness.

\subsection{Poverty trap without wariness}
\label{poverty-nowariness}
We will investigate the poverty trap in the absence of wariness. Let us prepare our results by an intermediate step.
\begin{lemma}[Proposition 1.3 in \cite{cm02}]
\label{gbeta}Let Assumptions \ref{HU} and \ref{assum2} be satisfied. Assume that $cu'(c)$ is increasing in $c$. Let $\beta>0$.

Then $nk_{t+1} = s_{\beta}(\omega(k_t), f'(k_{t+1}))$ is equivalent to $k_{t+1}=g_{\beta}(k_t)$ where $g_{\beta}: \rr_+\to \rr_+$ is  continuously differentiable, strictly increasing.

By consequence, the equilibrium capital path $k_t$ is unique and converges.
\end{lemma}
\begin{proof}See Appendix \ref{poverty-nowariness-proof}. \end{proof}

The following assumption is a key not only for the existence of a poverty trap but also for the comparative statics.
\begin{assum}\label{assum-poverty}The function $g_{\beta}$ in Lemma \ref{gbeta} satisfies the following conditions: (1) for $\beta_1>\beta_2>0$, $g_{\beta_1}(x)>g_{\beta_2}(x)$ $\forall x\geq 0$, (2) the set of fixed points $\mathcal{S}\equiv \{x>0: x=g_{\beta}(x)\}$ is non empty, and (3) $\lim_{x\to 0}\frac{g_{\beta}(x)}{x}<1$. 
\end{assum}


It should be noticed that Assumption \ref{assum-poverty}  holds under some standard setups, for instance, under CES production  and logarithm utility functions as we will prove later.\footnote{However, under Cobb-Douglas production functions, Assumption \ref{assum-poverty} may not hold. Indeed, if $u(c)=ln(c)$ and $f(k)=Ak^{\alpha}$, we have a dynamics $nk_{t+1}=\frac{\beta}{1+\beta}(1-\alpha)Ak^{\alpha}$, which is similar to the standar Solow model.  In this case, it is clear that $
g_{\beta}'(0)=\infty$ which violates condition $\lim_{x\to\infty}\frac{g_{\beta}(x)}{x}<1$ in Assumption \ref{assum-poverty}. See \cite{pp24} for detailed analysis regarding the effect of wariness.}.  

Assumption \ref{assum-poverty} leads to the following result.
\begin{lemma}\label{poverty-function-g}
Let Assumptions \ref{HU} and \ref{assum2} be satisfied. Assume that $cu'(c)$ is increasing in $c$ and Assumption \ref{assum-poverty} holds.  Then, there exists $x_{\beta}>0$, which is decreasing in $\beta$ such that $x_{\beta}=g_{\beta}(x_{\beta})$ and $g_{\beta}(x)<x$ $\forall x \in (0,x_{\beta})$.
\end{lemma}
\begin{proof}See Appendix \ref{poverty-nowariness-proof}. \end{proof}

\begin{proposition}
\label{comparativebeta}
Let Assumptions \ref{HU}, \ref{assum2} be satisfied. Let $\beta>0$.
Consider the economic system $nk_{t+1} = s_{\beta}(\omega(k_t), f'(k_{t+1}))$. Assume that $cu'(c)$ is  increasing in $c$ and Assumption \ref{assum-poverty} holds. Then, there exists $x_{\beta}>0$, which is decreasing in $\beta$ such that $x_{\beta}=g_{\beta}(x_{\beta})$ and $\lim_{t\to\infty}k_t=0$ for any $k_0 \in (0,x_{\beta})$. It means that $x_{\beta}$ is a poverty trap. Moreover, $x_{\beta}$ is not stable.
\end{proposition}
\begin{proof}See Appendix \ref{poverty-nowariness-proof}. \end{proof}

Definition 1.7 in \cite{cm02} mentions the notion of catching point: we says that $0$ is a catching point if for $g_{\beta}(k)<k$ for $k$ small. Then, \cite{cm02}'s Section 1.6.3 provides some necessary and sufficient conditions under which $0$ is a catching point. Here, our added value is to show not only the existence but also the monotonicity of $x_{\beta}$ which will be useful for studying the effects of warness on poverty trap.

We now consider a special case to illustrate and complement Proposition \ref{comparativebeta}.
 \begin{lemma}[No warriness]	\label{CesFull}Consider a CES production function and logarithm utility as in Corollary \ref{corollary-ces} and $\gamma = 0$. For $\beta>0$, the dynamical system $ nk_{t+1} = s_{\beta}(\omega(k_t), f'(k_{t+1})),\forall t$, becomes $k_{t+1}= g_{\beta}(k_t)\equiv 			\frac{\beta A(1-a)(ak_t^{\rho}+1-a)^{\frac{1}{\rho}-1}}{n(1+\beta)},$ where $g_{\beta}$ is defined by (\ref{func_g}), i.e., $g_{\beta}(x)= \frac{\beta A(1-a)(ax^{\rho}+1-a)^{\frac{1}{\rho}-1}}{n(1+\beta)}$, and $\omega(k)\equiv f(k)-kf'(k)=A(1-a)(ak^{\rho}+1-a)^{\frac{1}{\rho}-1}$.

	\begin{enumerate}
\item\label{ces-ex2-case1} If $\rho>0$, then $k_t$ converges to a strictly positive value for any $k_0>0$.

\item \label{ces-ex2-case2} If $\rho<0$, then $\max_{x\geq 0}\frac{\omega(x)}{x}\equiv -A\rho a^{\frac{1}{\rho}}(1-\rho)^{\frac{1}{\rho}-1}$, which achieves at $x_0\equiv \left(\frac{1-a}{-a\rho}\right)^{\frac{1}{\rho}}$.
\begin{enumerate}
	\item\label{ces-ex2-case2a}
	If $ -A\rho a^{\frac{1}{\rho}}(1-\rho)^{\frac{1}{\rho}-1}<\frac{n(1+\beta)}{\beta}$, then we have $\lim \limits_{t\to \infty}k_t=0$ for any $k_0>0$.
	
	\item \label{ces-ex2-case2b}
	If $ -A\rho a^{\frac{1}{\rho}}(1-\rho)^{\frac{1}{\rho}-1}=\frac{n(1+\beta)}{\beta}$, then $g_{\beta}(x_0)=x_0$ and 	\begin{align}
		&\lim_{t\to \infty} k_t=0 \text{ for any } k_0< x_0, 		&\lim_{t\to \infty} k_t&=x_0 \text{ for } k_0\geq x_0.
	\end{align}
	\item\label{ces-ex2-case2c}
	If $\max_{x\geq 0}\frac{\omega(x)}{x}\equiv -A\rho a^{\frac{1}{\rho}}(1-\rho)^{\frac{1}{\rho}-1}>\frac{n(1+\beta)}{\beta}$, then there exists $x_1, x_2$ such that $0<x_1<x_0<x_2$, $x_i=g_{\beta}(x_i)$ for $i=1,2$. Moreover,
\begin{align*}	&\lim_{t\to \infty}k_t=0 \text{ } \forall k_0< x_1, 		&\lim_{t\to \infty}k_t&=x_1 \text{ if } k_0= x_1, 		&\lim_{t\to \infty}k_t&=x_2 \text{ } \forall k_0 > x_1.\end{align*}

\end{enumerate}
\item \label{ComparativeStatics}
{\bf Comparative statics}. $x_1$ in point \ref{ces-ex2-case2c} is decreasing in $\beta,A$ but increasing in $a$.

Role of elasticity of factor substitution: Denote $y_s<1$ the second solution to the equation $B(y)\equiv (ay+1-a)\ln(ay+1-a)-(1-\rho)ay\ln(y)=0$ (this equation has two solutions $y_s$ and $1$). 

\begin{enumerate}
    \item \label{cesrho1} $x_1$ is decreasing in $\rho$ 
(i.e., $x_1'(\rho)<0$) iff $x_1^{\rho}\in  (0,y_s)\cup (1,\infty)$ (or equivalently, $x_1\in (0,1)\cup (y_s^{\frac{1}{\rho}}, \infty)$).
This happens when $x_1<x_0<1$, which is satisfied if the elasticity of factor substitution $\frac{1}{1-\rho}$ is quite high in the sense that $\frac{1}{1-\rho}>1-a$.

\item \label{cesrho2} $x_1$ is increasing in $\rho$ (i.e., $x_1'(\rho)>0$) iff $x_1^{\rho}\in (y_s,1)$ (or, equivalently, $1<x_1<y_s^{\frac{1}{\rho}}$).
\end{enumerate}
\end{enumerate}

 \end{lemma}
\begin{proof}See Appendix \ref{specialcase}.\end{proof}

 \cite{cm02}, pages 31-33, provide qualitative analyses to explain the existence of steady states $x_0,x_1,x_2$. We contribute by providing the global dynamics of $(k_t)$ in all cases and explicitly compute $\max_{x\geq 0}\frac{\omega(x)}{x}\equiv -A\rho a^{\frac{1}{\rho}}(1-\rho)^{\frac{1}{\rho}-1}$, which achieves at $x_0$. Another added value is that we show comparative statics to understand the effects of discount rate $\beta$ and substitutability parameters $a,\rho$.

As a direct consequence of Proposition \ref{comparativebeta} and Lemma \ref{CesFull}, the critical threshold $x_1$ in case \ref{ces-ex2-case2c} (or $x_0$ in case \ref{ces-ex2-case2b}) in Lemma \ref{CesFull} is the maximum poverty trap.

Some comments deserve mention with respect to the role of parameters $a$ and $\rho$, which characterize the substitutability of factors.
\begin{enumerate}
    \item 
    {\bf Capital intensity (automation) and poverty trap}.     According \cite{Aghionetal19}, the capital intensity parameter $a$ can be interpreted as the fraction of goods that have been automated by machines. Since $x_1$ in case \ref{ces-ex2-case2c} in Lemma \ref{CesFull} is increasing in $a$, our result indicates that the higher this fraction, the larger the set of poverty traps. This implies that when more goods can be automated, we need more capital to avoid poverty trap.
    \item {\bf Elasticity of factor substitution $\frac{1}{1-\rho}$ and poverty trap}.  Lemma \ref{CesFull}'s point  \ref{ComparativeStatics} shows that the effects of the elasticity of substitution on the set of poverty traps can be negative or positive depending on the economy's structure. 
    
According to the case (\ref{cesrho1}), when $\frac{1}{1-\rho}>1-a$ (i.e., $\rho>-\frac{a}{1-a}$), then the poverty trap $x_1$ is decreasing in $\rho$ and hence in the elasticity of factor substitution. So, the higher the elasticity of substitution, the smaller the set of poverty traps, the better chance to prevent poverty traps. Note also that case (\ref{cesrho1}) happens when $x_1^{\rho}<y_{s}$, i.e., the productivity is low (because $x_1^{\rho}$ is increasing in productivity $A$ and $y_s$ does not depend on $A$). According to  some empirical studies \citep{KlumpIrmen07, Knoblachetal2020}, the EFS would be less than $1$.  If we consider an example where the EFS equals $\frac{1}{1-\rho}=0.6$, then condition $\frac{1}{1-\rho}>1-a$ holds if $a>0.4$. 

However, in the case (\ref{cesrho2}), a higher elasticity of substitution generates more difficulty in escaping poverty traps. This happens when the productivity has an intermediate level (because $x_1^{\rho}$ is increasing in productivity $A$ and $y_s$ does not depend on $A$).

The literature, for example, \cite{Azariadis1996},  \cite{as05}, \cite{cm02},  provides some discussions regarding the role of elasticity of substitution on poverty traps. However, they did not explicitly provide comparative statics while we do. 
\end{enumerate}
\begin{ex}
This example illustrates our  comparative statics regarding the role the elasticity of substitution $\rho$ on the threshold of poverty trap $x_1$ in the standard OLG in Lemma \ref{CesFull}. With $A= 6.6, \beta = 0.75, n= 1.05, a= 0.35$, $x_1$ is decreasing in $\rho$ (see the graph on the left in Figure \ref{poverty_elasticity}). Inversely with $A= 6.6, \beta = 0.75, n= 1.05, a= 0.65$, $x_1$ is increasing in $\rho$ for $\rho \in [-2, -0.8]$ (see the graph on the right in Figure \ref{poverty_elasticity}). 
\begin{figure}[h]
    \centering
    \includegraphics[width=0.45\linewidth]{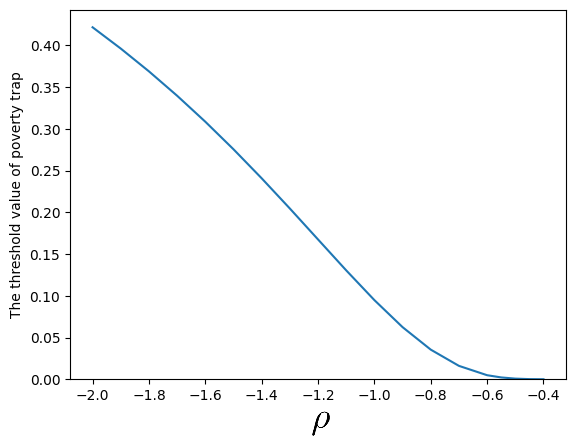}    \includegraphics[width=0.45\linewidth]{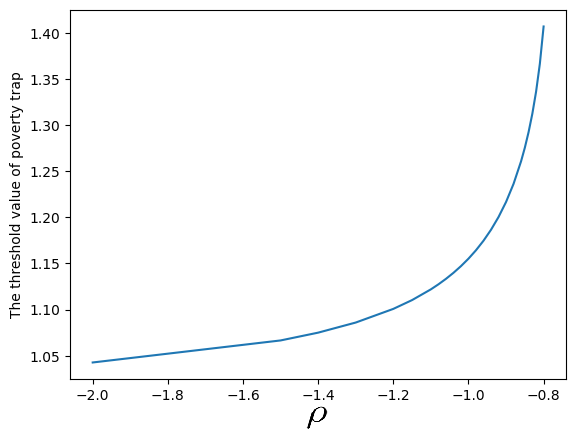}

    \caption{Effect of $\rho$ on the threshold of poverty trap $x_1$ in  Lemma \ref{CesFull}'s point  \ref{ComparativeStatics}.}
    \label{poverty_elasticity}
\end{figure}
\end{ex}


\subsection{Poverty trap with low wariness}\label{trap-wariness}

We now show the impact of wariness on the poverty trap in different circumstances. In the case of low productivity and low wariness, Proposition \ref{comparativebeta} leads to the following result.
\begin{proposition}[Low productivity and low wariness]\label{poverty2_1}Let Assumptions \ref{HU}, \ref{assum2}, \ref{assum-poverty} be satisfied and $cu'(c)$ be increasing in $c$.

If $f'(0)<\frac{1}{\beta_1}$, then the equilibrium capital path satisfies $nk_{t+1} = s_{\beta_1}(\omega(k_t), f'(k_{t+1}))$ $ \forall t$. By consequence, the equilibrium capital path $k_t$ is uniquen monotonic, and converges.

$x_{\beta_1}$ is a poverty trap. Moreover, it is decreasing in the wariness level $\gamma$ and $x_{\beta_1}<x_{\beta}$.

\end{proposition}
\begin{proof}See Appendix \ref{trap-wariness-proof}.\end{proof}



We complement Proposition \ref{poverty2_1} by focusing on a special case. The following result is a consequence of Proposition \ref{poverty2_1} and Lemma \ref{CesFull}.
 
 \begin{proposition}[Low productivity and low warriness]\label{cesll}
Consider CES production and logarithmic utility functions as in Corollary \ref{corollary-ces}.
Suppose that $\rho<0$ and $f'(0)=Aa^{\frac{1}{\rho}}<\frac{1}{\beta_1}\equiv \frac{1}{\beta +\gamma}$. We have $k_{t+1} = g_{\beta_1}(k_t)$. Recall that $\max_{x\geq 0}\frac{\omega(x)}{x}\equiv -A\rho a^{\frac{1}{\rho}}(1-\rho)^{\frac{1}{\rho}-1}$.
\begin{enumerate}
	\item \label{cesll1}
	If $\max_{x\geq 0}\frac{\omega(x)}{x}<\frac{n(1+\beta_1)}{\beta_1}$, then we have $\lim \limits_{t\to \infty}k_t=0$ for any $k_0>0$.
	
	\item \label{cesll2}
	If $\max_{x\geq 0}\frac{\omega(x)}{x}=\frac{n(1+\beta_1)}{\beta_1}$, then $g_{\beta_1}(x_0)=x_0$ and
\begin{align*}   
 \lim_{t\to \infty}k_t&=0 \text{ for any } k_0< x_0, 		&\lim_{t\to \infty}k_t &=x_0 \text{ for } k_0\geq x_0.
 \end{align*} 
    
	\item \label{cesll3}
	If $\max_{x\geq 0}\frac{\omega(x)}{x}>\frac{n(1+\beta_1)}{\beta_1}$, then there exists $x_1(\gamma), x_2(\gamma)$ such that $0<x_1(\gamma)<x_0<x_2(\gamma)$, $x_i(\gamma)=g_{\beta_1}(x_i(\gamma))$ for $i=1,2$. Moreover, 
 \begin{align*}	&\lim_{t\to \infty}k_t=0 \text{ } \forall k_0<   x_1(\gamma), 		&\lim_{t\to \infty}k_t&=x_1(\gamma) \text{ if } k_0=   x_1(\gamma), 		&\lim_{t\to \infty}k_t&=x_2(\gamma) \text{  } \forall k_0 > x_1(\gamma). 
 \end{align*}

\item 
{\bf Comparative statics}. $x_1(\gamma)$ in point \ref{cesll3} is decreasing in $\beta$ and $\gamma$ but increasing in the capital intensity $a$. However, as Lemma \ref{CesFull}'s point  \ref{ComparativeStatics}, $x_1(\gamma)$ can be increasing or decreasing in $\rho$ (and by consequence, the elasticity of substitution $\frac{1}{1-\rho}$).
\end{enumerate}

 \end{proposition}

 \paragraph{Comparative statics (under low productivity and low warriness).}

In the case (\ref{cesll1}), the productivity is very low, the economy collapses. 
Since the case \ref{cesll2} is not generic, let us focus on the case \ref{cesll3} which is the most interesting case in Proposition \ref{cesll}.  Here, $x_1$ is the maximum poverty trap, and the set of poverty traps is $[0,x_1]$.

By Lemma \ref{CesFull}'s point  \ref{ComparativeStatics}, we see that $x_1$ decreases in $\beta_1$ and so in the wariness level $\gamma$. Consequently, both Propositions \ref{poverty2_1} and  \ref{cesll} show that the higher the level of wariness, the smaller the set of poverty traps. It means that wariness has a positive effect on the dynamics of the economy in the case of low productivity and low wariness.



Next, we focus on the case of high productivity and low wariness.
\begin{proposition}[High productivity and low wariness]\label{poverty2}Let Assumptions \ref{HU}, \ref{assum2}, \ref{assum-poverty} be satisfied. Assume also that $cu'(c)$ is strictly increasing in $c$.

If $f'(\infty)>\frac{1}{\beta_2}$, then  the equilibrium capital path satisfies $ nk_{t+1} = s_{\beta_2}(\omega(k_t), f'(k_{t+1}))$ $\forall t$.  By consequence, the equilibrium capital path $k_t$ is unique, monotonic, and converges.
 
Moreover, the poverty trap $x_{\beta_2}$ is increasing in the wariness level $\gamma$ and $x_{\beta_2}>x_{\beta}$. 

\end{proposition}

\begin{proof}See Appendix \ref{trap-wariness-proof}.\end{proof}

We complement Proposition \ref{poverty2} by focusing on a special case.

 \begin{proposition}[High productivity and low warriness]\label{ces_l2}
Assume that $u(c)=ln(c)$ and $f(k)=A(ak^{\rho}+1-a)^{\frac{1}{\rho}}+Bk$, where $1-B\in [0,1]$ is the depreciation rate. We always have $\omega(k)\equiv f(k)-kf'(k)=A(1-a)(ak^{\rho}+1-a)^{\frac{1}{\rho}-1}$ as in the CES case  (Lemma \ref{CesFull}).

Assume that $\rho<0$ and $\beta B>1+\gamma$ (i.e., $f'(\infty)>1/\beta_2$).

The dynamics of  capital path is given by $k_{t+1}=g_{\beta_2}(k_t)$. 

\begin{enumerate}
	\item 
	If $\max_{x\geq 0}\frac{\omega(x)}{x}<\frac{n(1+\beta_2)}{\beta_2}$, then we have $\lim \limits_{t\to \infty}k_t=0$ for any $k_0>0$.
	
	\item 
	If $\max_{x\geq 0}\frac{\omega(x)}{x}=\frac{n(1+\beta_2)}{\beta_2}$, then $g_{\beta_1}(x_0)=x_0$ and 
\begin{align*}   
 \lim_{t\to \infty}k_t&=0 \text{ for any } k_0<   x_0, 		&\lim_{t\to \infty}k_t &=x_0 \text{ for } k_0\geq   x_0.
 \end{align*} 
    
	\item 
	If $\max_{x\geq 0}\frac{\omega(x)}{x}>\frac{n(1+\beta_2)}{\beta_2}$, then there exists $x_1=x_1(\beta_2), x_2=x_2(\beta_2)$ such that $0<x_1<x_0<x_2$, $x_i=g_{\beta_2}(x_i)$ for $i=1,2$. Moreover,
 \begin{align*}	&\lim_{t\to \infty}k_t=0 \text{ } \forall k_0< x_1, 		&\lim_{t\to \infty}k_t&=x_1 \text{ if } k_0= x_1, 		&\lim_{t\to \infty}k_t&=x_2 \text{ } \forall k_0 > x_1. 
 \end{align*}
\end{enumerate}
Note that $x_1(\beta_2)$ is the maximum poverty trap. Moreover, $x_1(\beta_2)$ is increasing in the wariness level $\gamma$.

 \end{proposition}

\begin{proof}See Appendix \ref{trap-wariness-proof}.\end{proof}

\paragraph{Comparative statics (under high productivity and low warriness).} Since $x_1(\beta_2)$ is increasing in the wariness level $\gamma$, Proposition \ref{poverty2} and Corollary \ref{ces_l2} show that wariness has a negative effect on the dynamics of the economy in the case of high productivity and low wariness.

\cite{pp24}'s Proposition 3 shows that under low wariness and low capital return (respectively, high capital return), the steady-state capital stock is increasing (respectively, decreasing) in the wariness level.  By the way, our results on the effect of wariness on poverty traps are consistent with those on the effects of wariness on economic growth in  \cite{pp24}. Our contribution is to explore the role of wariness and factor substitutability on poverty traps.

\subsection{Poverty trap with high wariness}
\label{highwariness}
Assume that the wariness level is high in the sense that $\frac{1}{\beta_1} \equiv \frac{1}{\beta+\gamma} \leq f'(\infty )< f'(0)\leq \frac{1}{\beta_2}\equiv \frac{1+\gamma}{\beta}$. In this case, the dynamics of capital becomes \begin{align}
nk_{t+1}=\frac{\omega(k_t)}{1+f'(k_{t+1})}, \forall t\geq 0, \text{ with $k_0>0$ is given.}
\end{align}
In Section \ref{MultipleEquilibrium}, we have shown that this may lead to multiple equilibria. The following result shows that a poverty trap may arise.
\begin{proposition}
    \label{poverty-strongcondition} 
Let Assumptions \ref{HU} and \ref{assum2} be satisfied. Assume that $\frac{1}{\beta_1} \equiv \frac{1}{\beta+\gamma} \leq f'(\infty )< f'(0)\leq \frac{1}{\beta_2}\equiv \frac{1+\gamma}{\beta}$.
	If $\lim_{k\to 0}\frac{\omega(k)}{k(1+f'(0))}<n$,	then there exists a poverty trap $\bar k$.
\end{proposition}
\begin{proof}See Appendix \ref{highwarinessProofs}.\end{proof}

We are now interested in studying comparative statics. For this purpose, 
we again focus on the case of logarithmic utility and CES production functions as in Corollary \ref{corollary-ces}. Let us introduce the function $H:\rr_+\to \rr_+$ by \textcolor{blue}{
\begin{align}
H(k)= \frac{\omega(k)}{h(k)}= \frac{\omega(k)}{k(1+f'(k))} =\frac{A(1-a)}{k(ak^\rho +1-a) ^{1-\frac{1}{\rho}} + aA k^{\rho}}.\end{align}}
\begin{proposition}\label{lem_r3}
	Consider a CES production function and logarithmic utility as in Corollary \ref{corollary-ces}. Suppose that conditions in Lemma \ref{h_increase}'s point \ref{hincreasing} hold (so that $h(k)\equiv k(1+f'(k))$ is increasing).
Consider the case $\gamma = \infty$. We have $0=\frac{1}{\beta_1}<f'(\infty)<f'(0)<\frac{1}{\beta_2}=\infty.$ and then the dynamics of capital is given by \begin{align*}
		nk_{t+1} =	\frac{A(1-a)(ak_t^{\rho}+1-a)^{\frac{1}{\rho}-1}}{1+f'(k_{t+1}) } , ~~~~ \forall t.
	\end{align*}
	\begin{enumerate}
		\item In the case that $0<\rho<1$, the capital path converges to a steady state which is the unique solution of the equation $\frac{\omega(k)}{h(k)}=n$.
		\item In the case that $\rho<0$, the following statements hold.
		\begin{enumerate}
			\item If $\max_{k\geq 0}\frac{\omega(k)}{h(k)}<n$ then $\displaystyle \lim\limits_{t\rightarrow \infty }k_t =0$ for any $k_0>0$.
			\item If $\max_{k\geq 0}\frac{\omega(k)}{h(k)}=n$, then there exists a unique $\bar{k}>0$ such that $\frac{\omega(\bar{k})}{h(\bar{k})}=n$. Furthermore, $\displaystyle \lim\limits_{t\rightarrow \infty }k_t =0$ for all $k_0<\bar{k}$ and $\displaystyle \lim\limits_{t\rightarrow \infty }k_t =x_0$ for all $k_0\geq \bar{k}$. By consequence, $\bar{k}$ is the maximum poverty trap.
			\item \label{HighWarinessCase3} If $\max_{k\geq 0}\frac{\omega(k)}{h(k)}>n$, then there exist two fixed points $\bar k_1, \bar k_2$ (i.e., $\frac{\omega(\bar k_i)}{h(\bar k_i)}=n$) and $0<\bar k_1<x^*<\bar k_2$, where $x^*$ is the unique solution to the following equation
$$1-a+\rho ax^{\rho}+\rho Aax^{\rho-1} (ax^{\rho}+1-a)^{\frac{1}{\rho}}=0.$$
We have (1) $\lim\limits_{t\rightarrow \infty }k_t =0$ for all $k_0<\bar k_1$, (2) $\lim\limits_{t\rightarrow \infty }k_t =\bar k_2$ for all $k_0>\bar k_1$, (3) $k_t=\bar{k}_1$ $\forall t$ if $k_0=\bar{k}_1$. By consequence, $\bar{k}_1$ is the maximum poverty trap. Moreover,  $\bar{k}_1$ is increasing in the capital intensity $a$ and decreasing in the productivity $A$. 
\end{enumerate}
	\end{enumerate}
\end{proposition}

The set of poverty traps is $[0,\bar{k}_1)$. 
Again we observe that the productivity $A$ plays an important role of the set of poverty traps: the higher the productivity, the smaller the set of poverty traps, the better chance to avoid a poverty trap.

\begin{proof}See Appendix \ref{highwarinessProofs}.\end{proof}

We illustrate the results in Corollary \ref{lem_r3} by a simulation.

\begin{ex}[$U(c,d)=\min(u(c),u(d))$]{\normalfont 
Let $\gamma=\infty$ and $n = 1.1,	a = 0.3, 	\rho = -0.6$.
\begin{enumerate}
    \item 	
	\begin{figure}[h]
		\centering
		\includegraphics[width=0.45\linewidth]{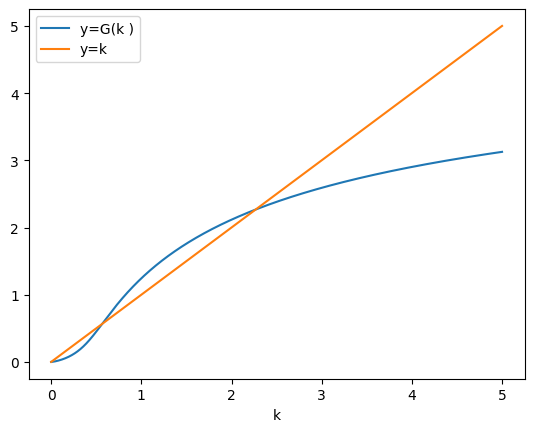}
		\includegraphics[width=0.45\linewidth]{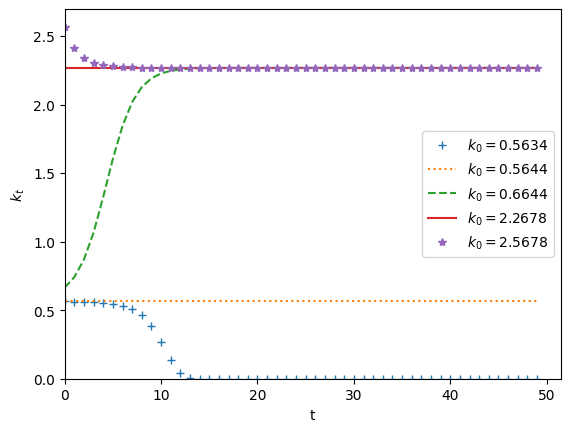}
		\caption{Dynamics of capital with $n = 1.1, 
			A = 3.6,
			a = 0.3,
			\rho = -0.6, \gamma =\infty $. The right-hand side of the figure shows the capital paths with different initial values $k_0$.}
		\label{regime41}
	\end{figure}
    For $A = 3.6$, the function $h(k)$ is increasing on $(0,\infty )$ (so, by Proposition \ref{unique-convergence} and Lemma \ref{h_increase}, there exists a unique equilibrium and it converges). There are two positive steady states $\bar k_1 \approx0.5644$ and $\bar k_2\approx 2.26776$. For any initial capital $k_0<\bar k_1$, the capital path converges to $0$. In contrast, if the initial capital $k_0>\bar k_1$ then the capital path converges to the higher steady state $\bar k_2$ (see Figure \ref{regime41}).

\item For $A = 2.973$, the function $h(k)$ is increasing on $(0,\infty )$. The capital path has only one positive steady state $\bar k \approx 1.06726$. For any initial capital $k_0<\bar k$, the capital path converges to $0$. Conversely, if the initial capital $k_0>\bar k_1$ then the capital path converges to the steady state $\bar k$ (see Figure \ref{regime51}).
	\begin{figure}[h]
		\centering
		\includegraphics[width=0.45\linewidth]{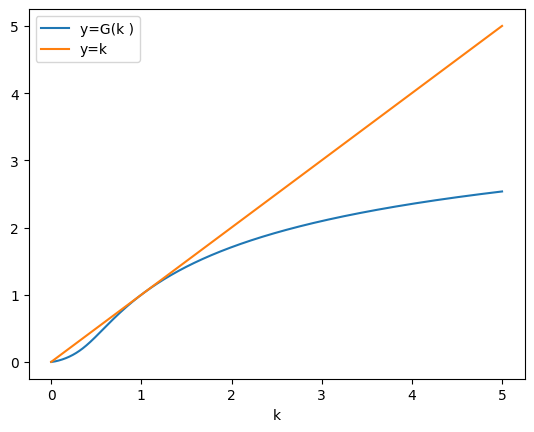}
		\includegraphics[width=0.45\linewidth]{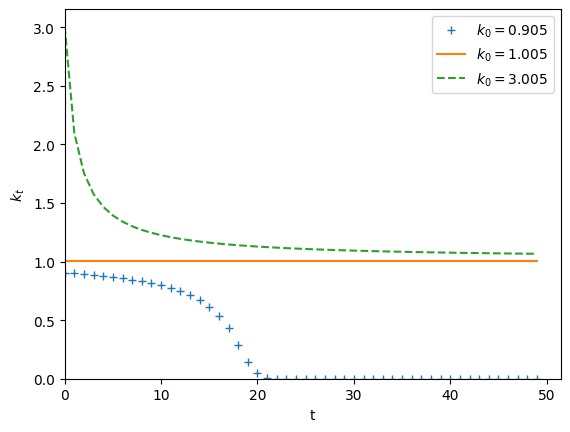}
		\caption{Dynamics of capital with $
			n = 1.1,
			A = 2.973,
			a = 0.3,
			\rho = -0.6, \gamma =\infty $. The right-hand side of the figure shows the capital paths with different initial values $k_0$.}
		\label{regime51}
	\end{figure}

\item	For $A = 2$, the function $h(k)$ is increasing on $(0,\infty )$. There is no positive steady state. The capital path converges to $0$ for any initial capital $k_0$ (see Figure \ref{regime61}).	
	\begin{figure}[h]
		\centering
		\includegraphics[width=0.45\linewidth]{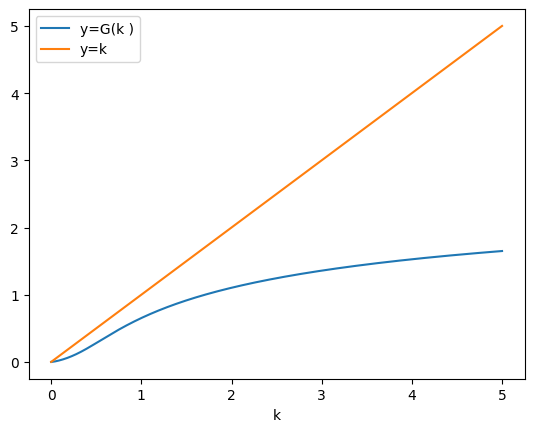}
		\includegraphics[width=0.45\linewidth]{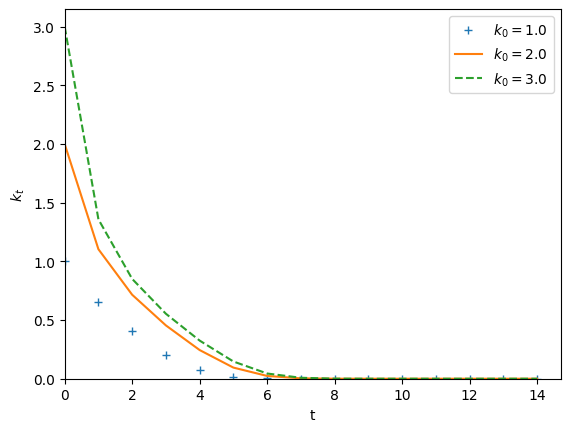}
		\caption{Dynamics of capital with
			$n = 1.1,
			A = 2,
			a = 0.3,
			\rho = -0.6, \gamma =\infty$. The right-hand side of the figure shows the capital paths with different initial values $k_0$.}
		\label{regime61}
	\end{figure}
    \end{enumerate}
    }
\end{ex}

\subsection{Poverty trap with intermediate wariness and productivity}
\label{intermidiatesection}
In this section, we study the case where wariness and productivity have a middle level. It should be noticed that in this case, the capital stock $k_t$ may not be  in the same regime  (Definition \ref{regimedefinition}) across periods. However, we manage to prove the existence of a poverty trap in such a general setting.
\begin{proposition}[poverty trap with intermediate productivity and  wariness]\label{poverty-middle}

Let Assumptions \ref{HU}, \ref{assum2}, \ref{assum-poverty} be satisfied. Assume that $cu'(c)$ is  increasing in $c$.

Assume that  $f'(\infty)<\frac{1}{\beta+\gamma} <\frac{1+ \gamma}{\beta}<f'(0)$.\footnote{For CES production function with $\rho<0$, we have $f'(0)=Aa^{\frac{1}{\rho}}$ and $f'(\infty)=0$.} Let $x_{\beta_1}$ be defined in Assumption \ref{assum-poverty}.

For $k_0>0$ be small enough in the sense that 
\begin{align}
k_0&<{x_{\beta_1}}\\
\label{omega(k0)}\omega(k_0)&<n\big(1+\frac{1}{\gamma+\beta}\big)
\left(f'\right)^{-1}\big(\frac{1+\gamma}{\beta}\big),
\end{align}
 then the equilibrium capital converges to zero.  By consequence, the threshold 
\begin{align}\label{x_poverty_threshold}
 x_{poverty} = \min\left(x_{\beta_1}, \omega^{-1} \big(
n\big(1+\frac{1}{\gamma+\beta}\big)
\big)(f^{\prime})^{-1}\big(\frac{1+\gamma}{\beta}\big) \right)
\end{align} is a poverty trap. Moreover,  $x_{poverty}$ is decreasing in wariness level $\gamma$. 
	
\end{proposition}
\begin{proof}See Appendix \ref{intermidiatesection-proofs}.\end{proof}


Although $x_{poverty}$ in Proposition \ref{poverty-middle} is a poverty trap, it may not be the maximum value of poverty traps. It means that there may exist some other value of initial capital $k_0>x_{poverty}$, whose associated capital path converges to zero. The following example provides an illustration.
\begin{ex}{\normalfont Consider again logarithm utility and CES production functions. Taking the values $A = 3.3, 
a = 0.3, 
\rho = -0.9, n= 1.32, \beta= 0.7, \gamma = 0.54$, we have $f'(0) \approx 12.5744> \frac{1}{\beta_2} \approx 2.2 >\frac{1}{\beta_1}\approx 0.806>0$. The threshold in (\ref{x_poverty_threshold}) is $x_{poverty} \approx 0.0887$. If the initial value $k_0< x_{poverty}$, then the capital path converges to $0$. However,  there also exists a scenario that $k_0>x_{poverty}$ but the capital still decreases to $0$ (see Figure \ref{regime9}). 
\begin{figure}[h]
	\centering
	\includegraphics[width=0.5\linewidth]{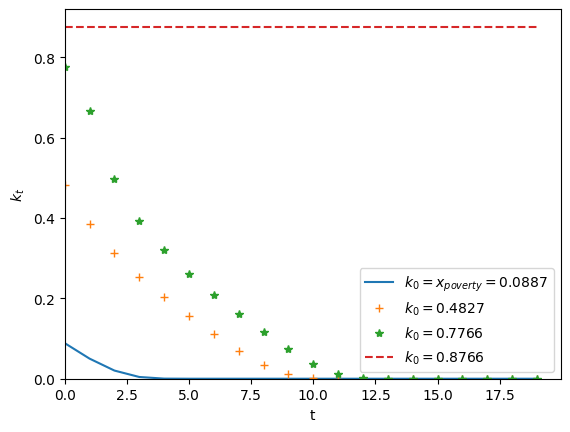}
	\caption{Capital paths with different initial values near $x_{poverty}$}
	\label{regime9}
\end{figure}
}
\end{ex}



We now complement Proposition \ref{poverty-middle} by showing more detailed results with the CES production function. Under this specification, we can explicitly identify the whose set of all poverty traps.

		\begin{proposition}\label{poverty_x1}
		Consider a CES production function  with $\rho<0$ and logarithm utility as in  Corollary \ref{corollary-ces}. Suppose that  the conditions in Lemma \ref{h_increase} are satisfied and
		\begin{align*}
		\max_{x\geq 0}\frac{\omega(x)}{x}= 	-A\rho a^{\frac{1}{\rho}}(1-\rho)^{\frac{1}{\rho}-1} \geq \frac{n(1+\beta_1)}{\beta_1}
		\end{align*}
		Let $x_{\beta_1}$ be the smallest positive solution of $g_{\beta_1}(k)=k$.   Assume that if $f'(x_{\beta_1})<\frac{1}{\beta_1}$ 
		and one of the following conditions satisfies: (i) $H(k)<k$ $\forall k>0$, (ii) $\bar k_1>x_{\beta_1}$ where $\bar k_1$ is the smallest positive solution to the equation $H(k)=k$.

Then, the following statements hold.
\begin{enumerate}
\item If $k_0<x_{\beta_1}$, then $\lim_{t\to\infty}k_t=0$.
\item If $k_0=x_{\beta_1}$, then $k_t=x_{\beta_1}$.
\item If $k_0>x_{\beta_1}$, then $\lim\limits_{t\to\infty}k_t\geq x_{\beta_1}$.    
\end{enumerate}
	\end{proposition} 
\begin{proof}See Appendix \ref{intermidiatesection-proofs}.\end{proof}
Here, the set of poverty traps is $[0,x_{\beta_1})$. Observe that, the maximum poverty trap $x_{\beta_1}$ is decreasing in the wariness level $\gamma$ and $x_{\beta_1}< x_{\beta}$. So, the wariness $\gamma$  helps to reduce the set of povery traps, which is consistent with the insight in Propositions \ref{poverty2_1} and  \ref{cesll}.

\subsection{Collapsing economy}
So far we have explained why a poverty trap may exist. However, under strong conditions, the economy collapses whatever the initial condition.
\begin{proposition}[Collapsing economy]
\label{steady0}
Let the assumptions of Proposition \ref{unique-convergence} hold and one of the following conditions holds:
\begin{enumerate}
\item\label{part1} 		$f'(\infty)\leq \frac{1}{\beta_1}\equiv \frac{1}{\beta+\gamma}< \frac{1}{\beta_2}\equiv \frac{1+\gamma}{\beta}\leq f'(0)$, $  M_1< n$ and  $ M_3<n$.

\item $f'(\infty)\geq  \frac{1+\gamma}{\beta_2}$ and $M_2< n		$.
\item $f'(0)\leq  \frac{1}{\beta+\gamma}$ and $ M_1< n	$.
\item $ \frac{1}{\beta +\gamma} \leq f'(\infty)< f'(0)\leq  \frac{1+\gamma}{\beta}$ and $ M_3<n$.	
	\end{enumerate}
Then the economy collapses (i.e., $\lim\limits_{t\rightarrow \infty}k_t=0$) 	for any $k_0>0$.
    
\end{proposition}

\begin{proof}See Appendix \ref{intermidiatesection-proofs}.\end{proof}

Recall that $M_i$ is defined by (\ref{notationMi}). To better understand Proposition \ref{steady0}, let us consider the production function $\tilde{f}\equiv Af(\cdot)$ instead of $f(\cdot)$, where $A$ represents the productivity.  With the obvious notations, we have $\tilde{M}_1=AM_1,\tilde{M}_2=AM_2$ and $\tilde{M}_2=A\sup_{k>0}\frac{\omega(k)}{k(1+Af'(k))}$. So, applying Proposition \ref{steady0} for the function $\tilde{f}$, we see that, in all cases, the economy collapses if the productivity $A$ is low enough.
		


\section{Conclusion}\label{conclu}
We have investigated the effects of wariness on poverty traps and equilibrium multiplicity. We have shown that whether wariness could reduce or increase the possibility of the poverty trap depends on the interaction between productivity, wariness, and factor substitutability.   Interestingly, under low productivity and low wariness, wariness has a positive effect on the dynamics of the economy and can reduce the set of poverty traps because it enhances investment. However, under high productivity and low wariness, wariness has a negative effect.

We have also proved that a high level of wariness can generate an equilibrium multiplicity. With the same fundamentals, there may exist an equilibrium that exhibits a collapsing behavior (i.e., the capital path converges to zero) and another equilibrium whose capital path converges to a stable steady state.

We have highlighted the interplay between wariness, factor substitution (capital intensity, elasticity of factor substitution), and poverty traps. Under standard specifications, we show that the lower the capital intensity, the smaller the set of poverty traps, and, by the way, the higher chance to avoid poverty traps. However, the effect of the elasticity of substitution on povery traps is not necessarily monotonic. 

\subsubsection*{Acknowledgement}
This research is funded by Vietnam National University Ho Chi Minh City (VNU-
HCM) under grant number C2024-28-11.

\appendix
\section*{Appendix: formal proofs}

\section{Proofs for Section \ref{section3}}
\label{section3-proofs}
\begin{proof}[{\bf Proof of Proposition \ref{unique-convergence}}]
We follow the strategy of \cite{pp24}. Proposition \ref{unique-convergence} is a direct consequence of the two following lemmas.
\begin{lemma}\label{k_increasing}
Assume that the function $f^{\prime}(k)u^{\prime}\big(nkf^{\prime}(k)\big)$ is decreasing  is increasing on the interval $(0,\infty)$. Given $k_t>0$, let $k_{t+1,1}$ be determined by the Euler equation 
$u'\left(\omega(k_t) - nk_{t+1,1}\right) = \beta_1 R_{t+1}u'\left(R_{t+1}nk_{t+1,1}\right)$ and $R_{t+1}=f'(k_{t+1,1}) < \frac{1}{\gamma+\beta}$, where recall that $\omega(k)\equiv f(k_t)-k_tf^{\prime}(k_t)$.   Then $k_{t+1,1}$ is a strictly increasing, continuous function of $k_t$. 


\end{lemma}

\begin{proof}
It suffices to prove that $w_t$ is increasing in $k_{t+1}$. 
Taking the derivative of both sides of the equation $u'\left(\omega(k_t) - nk_{t+1}\right) = \beta_1 R_{t+1}u'\left(R_{t+1}nk_{t+1}\right)$ with respect to $k_{t+1}$ and noting that $R_{t+1}=f^{\prime}(k_{t+1})$, we have 
\begin{align*}
(\frac{\partial \omega(k_t)}{\partial k_{t+1}}-n)u^{\prime \prime}(c_t)=\beta_1\frac{\partial \big(R_{t+1}u'\left(R_{t+1}nk_{t+1}\right)\big)}{\partial k_{t+1}} \\
\text{or, equivalently, }\frac{\partial \omega(k_t)}{\partial k_{t+1}}u^{\prime \prime}(c_t)=nu^{\prime \prime}(c_t)+\beta_1\frac{\partial \big(R_{t+1}u'\left(R_{t+1}nk_{t+1}\right)\big)}{\partial k_{t+1}}
\end{align*}
Since $u^{\prime \prime}<0$ and $f^{\prime}(k)u^{\prime}\big(nkf^{\prime}(k)\big)$ is decreasing in $k$ for any $k>0$, we have $\frac{\partial \omega(k_t)}{\partial k_{t+1}}>0$. Note that $\omega(k_t)=f(k_t)-k_tf^{\prime}(k_t)$ is increasing in $k_t$ (because the function $f$ is strictly concave). By consequence, we get that $k_{t+1,1}$ is an increasing function of $k_t$.

\end{proof}

\begin{lemma}\label{prop3}
Assume that 
$f^{\prime}(k)u^{\prime}\big(nkf^{\prime}(k)\big)$ is decreasing and $h(k)\equiv k+kf'(k)$ is strictly increasing in $k$  for any $k>0$. 
 Then $k_{t+1}$ determined by (\ref{kt+1}) is a strictly increasing, continuous function, {\color{blue} denoted by $G(k_t)$,} of $k_t$. By consequence, the capital path $(k_t)$ converges.
 \end{lemma}
 \begin{proof}By using the same argument in Lemma \ref{k_increasing}, we can prove that $k_{t+1}$ determined by (\ref{kt+1}) is an increasing and continuous function of $k_t$ in each case in the formula (\ref{kt+1}). So, the function $G(k_t)$ is increasing and continuous.
\end{proof}

\end{proof}

\begin{proof}[{\bf Proof of Lemma \ref{h_increase}}]
			It is clear that for $0<\rho<1$, the function $kf'(k)$ is increasing in $k$ and so is $h(k)$.
			
We now consider the case $\rho<0$. We have
			\begin{align*}
				h'(k) =1 + kf''(k)  + f'(k) \text{ and }
				h''(k)= kf'''(k) + 2f''(k).
			\end{align*}
			It is easy to verify that  \begin{align*}
				\frac{f''(k)}{f'(k)}& = (\ln f'(k))' = -\frac{(1-\rho)(1-a)}{k(ak^\rho + 1-a)}\\
				h'(k) &= 1 - \frac{(1-\rho)(1-a)}{ak^\rho + 1-a}f'(k) + f'(k)\\
				h''(k)& = \frac{f''(k)}{ak^\rho + 1-a} \left[a(1-\rho)k^{\rho} + \rho(1-a)\right]
			\end{align*}
			The term  $\frac{f''(k)}{ak^\rho + 1-a}$ is always negative while $a(1-\rho)k^{\rho} + \rho(1-a)$ is decreasing with one root which is $x_c \equiv  \left(-\frac{\rho(1-a)}{a(1-\rho)}\right)^{\frac{1}{\rho}}$. Hence $h'(k)$ achieves minimum at $x_c$. The function $h(k)$ is increasing on $(0,\infty )$ if $h'(x_c)\geq 0$ which is equivalent to $1- Aa^{\frac{1}{\rho}} (-\rho)^{-\frac{1}{\rho}+3}(1-2\rho)^{\frac{1}{\rho}-2}\geq 0$.
		\end{proof}

\begin{proof}[{\bf Proof of Corollary \ref{prop_convergence}}]
Let us prove point 1 (we can use the same argument to prove points 2 and 3). It is clear that if $f'(k^*)<\frac{1}{\beta_1}$ then $k^*$ must satisfies the equation $nk =	s_{\beta_1}(\omega(k^*),  f'(k))$.  If $M_1<n  $ then the equation $nk =	s_{\beta_1}(\omega(k^*),  f'(k))$ has no positive solution. 

\end{proof}

\begin{proof}[{\bf Proof of Proposition \ref{wariness-multipleequilibria}}]
According to Proposition \ref{proposition1}'s part 2, we have $s_t=\frac{w_t}{1+R_{t+1}}$, $\forall t$. By consequence, according to Definition  \ref{def2}, $(k_t)$ is an equilibrium if and only if $
nk_{t+1}=\frac{\omega(k_t)}{1+f'(k_{t+1})}, \forall t\geq 0.$

Assume now that the equation  $\omega(k_0)=nk(1+f'(k))$ has at least two strictly positive solutions, denoted by $k_1^1, k_1^2$.

Consider the function $h(k)\equiv k+kf'(k)$. 

By the concavity of $f$, we have $kf'(k)\leq f(k)-f(0)$. So, $\lim_{k\to 0} kf'(k)=0$. Hence, $\lim_{k\to 0} k+kf'(k)=0$. It is easy to see that $\lim_{k\to 0} k+kf'(k)=\infty$. By consequence, for given $w>0$, there exists at least $k>0$ such that $nk=\frac{w}{1+f'(k)}$. 

Now, given $k_0$ and $k_1^i>0$ (i=1,2), there exists at least 1 equilibrium with $k_1=k_1^i$. Therefore, there are at least two equilibria.
\end{proof}

\begin{proof}[{\bf Proof of Proposition \ref{multiple-version2}}]
	If $nk = s_{\beta_1}(\omega(k_0), f'(k))$ has a solution $k^{(1)}_{1}$ which satisfies $f'(k^{(1)}_{1})<\frac{1}{\beta_1}$ then $(k_0, k_{1}^{(1)})$ is a temporal  equilibrium where the next capital $k_{1}^{(1)}$ belongs to the regime 1.  
	
	Similar if the condition 2 satisfies then $(k_0, k_{1}^{(2)})$ is a temporal  equilibrium where the next capital $k_{1}^{(2)}$ belongs to the regime 2.  If the last condition satisfies then we also have the same observation. 
	
	Due to the strictly decreasing property of $f'$, the $k^{(i)}_{1}$ for $i=1, 2, 3$ (if exists) must be different because the values of $f'$ at these points fall into non-overlapping intervals.  Hence if at least two of these conditions satisfies, there exists at least two equilibria.
\end{proof}

\section{Proofs for Section \ref{section4trap}}
\label{section4trap-proof}
\subsection{Proofs for Section \ref{poverty-nowariness}}
\label{poverty-nowariness-proof}

\begin{proof}[{\bf Proof of Lemma \ref{gbeta}}]
To make our paper self-contained, we present a simple proof.  By Definition \ref{sbeta} of the function $s_{\beta}(\omega, R)$, we have $u'\left(\omega- s\right) = \beta R  u'\left(R s\right).$ The function $s_{\beta}$ is continuously differentiable. Since the function $cu'(c)$ is  increasing in $c$, we have $\frac{\partial s_{\beta}}{\partial R}\geq 0$.

 Define the function $\Delta: \rr_+^{2}\to \rr$ by $\Delta(k,w)\equiv nk-s_{\beta}(\omega,f'(k))$.  We have 
\begin{align*}
\frac{\partial \Delta}{\partial k}(k,\omega)=n-\frac{\partial s_{\beta}}{\partial R}(\omega,f'(k))f^{''}(k)
\end{align*}

For $\omega>0$, let $k$ be the solution to the equation $\Delta(k,\omega)=0$. For any $\omega>0$, we have \begin{align*}
\frac{\partial \Delta}{\partial k}(k(\omega),\omega)=n-\frac{\partial s_{\beta}}{\partial R}(\omega,f'(k(\omega))f^{''}(k(\omega))\geq n>0
\end{align*}
because $\frac{\partial s_{\beta}}{\partial R}\geq 0$ and  $f^{''}(k)<0$.

So, by the implicit theorem, there exists a function $h: \rr_{++}\to \rr_{++}$  which is continuously differentiable so that $\Delta(k,\omega)=0 \Leftrightarrow k=h(\omega)$.

Then we define the function $g_{\beta}$ by $g_{\beta}(k)=h(\omega(k))$.
\end{proof}

\begin{proof}[{\bf Proof of Lemma \ref{poverty-function-g}}]

{\bf Step 1}. For $\beta>0$, denote $x_{\beta}\equiv \inf\{x\in \mathcal{S}\}.$
Obviously, we have $x_{\beta}=g_{\beta}(x_{\beta})$. By the assumption  $\lim_{x\to\infty}\frac{g_{\beta}(x)}{x}<1$, there exists an open set $(0,\epsilon)$ such that $g_{\beta}(x)<x$ for all $x\in (0,\epsilon)$. This implies that $x_{\beta}>0$. Indeed, if $x_{\beta}=0$, then, by the definition of $x_{\beta}$, there exists a sequence $x_n\in \mathcal{S}$ such that $x_n$ converges to zero. However, since $x_n$ converges to zero, there exists $n_0$ such that $x_n<\epsilon$ for any $n\geq n_0$. This is impossible because $g_{\beta}(x)<x$ for all $x\in (0,\epsilon)$. Therefore, we must have $x_{\beta}>0$.

{\bf Step 2}. we prove that $g_{\beta}(x)<x, \forall x \in (0,x_{\beta})$. Let $x\in (0,x_{\beta})$. By the definition of $x_{\beta}$, we cannot have $x=g_{\beta}(x)$ because $x\in (0,x_{\beta})$. If $x>g_{\beta}(x)$, then by the assumption $\lim_{x\to 0}\frac{g_{\beta}(x)}{x}<1$, there exists $x_1\in (0,x)$ such that $x_1=g_{\beta}(x_1)$. This is impossible by the definition of $x_{\beta}$. Therefore, we must have $x<g_{\beta}(x)$.

{\bf Step 3}. let $\beta_1 \geq \beta_2 >0$. We claim that $x_{\beta_{1}} \leq x_{\beta_{2}}$. Suppose $x_{\beta_{1}}>x_{\beta_{2}}$. Then, we have $x_{\beta_2}\in (0,x_{\beta_1})$. According to the step 2 above, we have $g_{\beta_1}(x_{\beta_2} )<x_{\beta_2}$. However, by Assumption \ref{assum-poverty} and $\beta_1>\beta_2$, we have $g_{\beta_1}(x_{\beta_2})>g_{\beta_2}(x_{\beta_2})=x_{\beta_2}$, a contradiction. By consequence, we have $x_{\beta_1}\leq x_{\beta_2}$.

\end{proof}

\begin{proof}[{\bf Proof of Proposition \ref{comparativebeta}}]
Let $k_0<x_{\beta}$. We have $k_1=g_{\beta}(k_0)<k_0<x_{\beta}$. By induction, we have $k_{t+1}<k_t<x_{\beta}$ for any $t$. So, $k_t$ converges to some value, say $k^*$. Since $k^*<k_0<x_{\beta}$ and $k^*=g_{\beta}(k^*)$, we must have $k^*=0$. Indeed, if $k^*>0$, we must have, by Assumption \ref{assum-poverty},  $g_{\beta}(k^*)<k^*$, a contradiction.

\end{proof}

	
	

\subsection{Proofs for Section \ref{trap-wariness}}
\label{trap-wariness-proof}

\begin{proof}[{\bf Proof of Proposition \ref{poverty2_1}}]
 If $f'(0)<\frac{1}{\beta_1}$,  Corollary \ref{basic1} implies that $ nk_{t+1} = s_{\beta_1}(\omega(k_t), f'(k_{t+1})),\forall t$. 
 
 According to Lemma \ref{gbeta}, this system is equivalent to $k_{t+1}=g_{\beta_1}(k_t)$ where $g_{\beta_1}$ is continuous, strictly increasing. By consequence, the equilibrium capital path $k_t$ is unique, monotonic, and converges.

By Lemma \ref{comparativebeta} and $\beta_1\equiv \beta+\gamma$, the threshold $x_{\beta_1}$ is a poverty trap and it is decreasing in the wariness level $\gamma$ and $x_{\beta_1}<x_{\beta}$.
\end{proof}

\begin{proof}[{\bf Proof of Proposition \ref{poverty2}}]
		
 If $f'(\infty)>\frac{1}{\beta_2}$,  Corollary \ref{basic1} implies that $ nk_{t+1} = s_{\beta_2}(\omega(k_t), f'(k_{t+1})),\forall t$. 
 
 According to Lemma \ref{gbeta}, this system is equivalent to $k_{t+1}=g_{\beta_2}(k_t)$ where $g_{\beta_2}$ is continuous, strictly increasing.
 
By Lemma \ref{comparativebeta} and $\beta_2\equiv \frac{\beta}{1+\gamma}$, the threshold $x_{\beta_2}$ is a poverty trap and it is increasing in the wariness level $\gamma$ and $x_{\beta_2}>x_{\beta}$.
\end{proof}

\begin{proof}[{\bf Proof of Proposition \ref{ces_l2}}]
    We also present some computations. \begin{align*}
	f(k)&=A(ak^{\rho}+1-a)^{\frac{1}{\rho}}+Bk, \quad 
	f'(k)=Aak^{\rho-1}(ak^{\rho}+1-a)^{\frac{1}{\rho}-1}+B\\	
		kf'(k)&=Aak^{\rho}(ak^{\rho}+1-a)^{\frac{1}{\rho}-1}+Bk\\
			\omega(k)&=f(k)-kf'(k)=A(1-a)(ak^{\rho}+1-a)^{\frac{1}{\rho}-1}.
	\end{align*}
We focus on the case $\rho<0$. We have $f(0^+) =0$, $f(\infty)=A(1-a)^{\frac{1}{\rho}}$ if $B=0$, $f(\infty)=\infty$ if $B>0$. $f'(0)=Aa^{\frac{1}{\rho}}+B$, $f'(\infty)=B$. Moreover,  $\omega(k)$ is increasing (because $f$ is concave). $\omega(0)=0,$ $\omega(\infty)=A(1-a)^{\frac{1}{\rho}}$ for $\rho<0$.

Since $\beta B>1+\gamma$, we have $f'(\infty)=B>\frac{1}{\beta_2}$, by Corollary \ref{basic1}, the equilibrium capital path satisfies $ nk_{t+1} = s_{\beta_2}(\omega(k_t), f'(k_{t+1}))$ $\forall t$.
However, under the logarithmic utility $u(c)=ln(c)$, we find that $s_{\beta_2}(\omega(k_t), f'(k_{t+1}))=\frac{\beta_2}{1+\beta_2}\omega(k_t)$. By consequence, $$k_{t+1}=\frac{\beta_2}{(1+\beta_2)n}A(1-a)(ak^{\rho}+1-a)^{\frac{1}{\rho}-1}=g_{\beta_2}(k_t).$$ Therefore, by applying Lemma \ref{CesFull} with $\beta$ is replaced by $\beta_2$, We obtain our results.
\end{proof}

\subsection{Proofs for Section \ref{highwariness}}
\label{highwarinessProofs}
\begin{proof}[{\bf Proof of Proposition \ref{poverty-strongcondition}}]
We have $\frac{k_{t+1}}{k_t}\leq \frac{\omega(k_t)}{nk_t(1+f'(k_{t+1}))}\leq \frac{\omega(k_t)}{nk_t(1+f'(0))}$ for all $k_t>0$. 
	
Condition $\lim_{k\to 0}\frac{\omega(k)}{k(1+f'(0))}<n$ implies that there exists $\bar k>0$ and $\gamma\in (0,1)$ such that for all $k<\bar k$, we have $\frac{\omega(k)}{nk(1+f'(0))}<\gamma$. This implies that $k_t$ converges to zero for any $k_0\in (0,\bar{k})$.
\end{proof}

\begin{proof}[{\bf Proof of Proposition \ref{lem_r3}}]
It is clear that the condition $\frac{1}{\beta_1}<f'(\infty)<f'(0)<\frac{1}{\beta_2}$ holds if $\gamma = \infty $. Thanks to point 4 in Corollary \ref{basic1}, the dynamics of capital is given by \begin{align*}
	nk_{t+1} =	\frac{A(1-a)(ak_t^{\rho}+1-a)^{\frac{1}{\rho}-1}}{1+f'(k_{t+1}) } , ~~~~ \forall t.
\end{align*}
	\begin{enumerate}
		\item If $\rho>0$ then $\frac{\omega(x)}{h(x)}$ is increasing. Hence the equation $\frac{\omega(x)}{h(x)}=n$ has a unique solution. 
		\item  Consider the case  $\rho<0$. Recall that   $           H(x)=\frac{A(1-a)}{x(ax^\rho +1-a) ^{1-\frac{1}{\rho}} + aA x^{\rho}}$.  We have $\lim\limits_{x\rightarrow \infty} H(x)= \lim\limits_{x\rightarrow 0} H(x) =0$.
Consider $b(x)\equiv aAx^{\rho}+x(ax^{\rho}+1-a)^{1-\frac{1}{\rho}}$. We have 
\begin{align*}
    b'(x)&=\rho Aax^{\rho-1}+(ax^{\rho}+1-a)^{1-\frac{1}{\rho}}+(\rho-1)ax^{\rho}(ax^{\rho}+1-a)^{1-\frac{1}{\rho}-1}\\
    &=(ax^{\rho}+1-a)^{\frac{-1}{\rho}}\Big((\rho-1)ax^{\rho}+ax^{\rho}+1-a+\rho Aax^{\rho-1} (ax^{\rho}+1-a)^{\frac{1}{\rho}}\Big)\\
    &=(ax^{\rho}+1-a)^{\frac{-1}{\rho}} \Big(1-a+\rho ax^{\rho}+\rho Aax^{\rho-1} (ax^{\rho}+1-a)^{\frac{1}{\rho}}\Big).
\end{align*}
Since $b(x)=\frac{A(1-a)}{H(x)}$, we have  $b'(x)=- H'(x)\frac{A(1-a)}{(H(x))^2}$. By consequence, 
		\begin{align*}
			H'(x)\frac{A(1-a)}{(H(x))^2}  =X^{-\frac{1}{\rho}}\left[-\rho a x^{\rho-1} \left(x+ AX^{\frac{1}{\rho}}\right)-(1-a)\right],
		\end{align*}
		where  $X = ax^\rho + 1-a$. Since $\rho<0$, 
we see that $-\rho a x^{\rho-1} \left(x+ A(ax^\rho + 1-a)^{\frac{1}{\rho}}\right)$ is increasing in $k$. Moreover, we observe that
		\begin{align*}
			&	\lim\limits_{x\rightarrow 0^+} \left[1-a + a\rho x^{\rho-1} \left(x+ A(ax^\rho + 1-a)^{\frac{1}{\rho}}\right )\right] = -\infty  \\
			&		\lim\limits_{x\rightarrow +\infty } \left[1-a + a\rho x^{\rho-1} \left(x+ A(ax^\rho + 1-a)^{\frac{1}{\rho}}\right )\right] =1-a>0
		\end{align*}
		Hence, there exists a unique $k^*>0 $ such that $H'(k^*) =0$, $H'(x)>0$ for $x<k^*$ and $H'(x)<0$ for $x>k^*$. In other words, $H$ achieves the maximum at $k^*$. 
 
        \begin{enumerate}
			\item If $M\equiv \max_{x\geq 0}\frac{\omega(x)}{h(x)}<n $, then  there is no positive steady state. Furthermore, 
			\begin{align*}
				\frac{h(k_{t+1})}{h(k_t)} = \frac{\omega(k_t)}{nh(k_t)}<\frac{M}{n}<1.
			\end{align*}
			So, $h(k_t)$ converges to zero and $h(k_{t+1})< h(k_{t})$ (which	implies that $k_{t+1} <{k_t}$ because $h$ is increasing). Thus $\lim\limits_{t\rightarrow \infty } k_t= 0$ for all initial capital $k_0$. 
			\item If $M=n$, then $\frac{\omega(k)}{h(k)} =n $ has a unique solution $k^*$. 
			\item If $M>n$, then the equation $\frac{\omega(k)}{h(k)}=n$ has two positive solutions $\bar k_1<\bar k _2$ with $\bar k_1<k^*<\bar k _2$.     

		\end{enumerate}
\item 
{\bf Comparative statics}. We focus on the case $\rho<0$ and $M>n$.

{\bf Role of $a$}. We claim that $\bar{k}_1$ is increasing in $a$.  Consider a fixed point $H(x)=n$, i.e., $b(x)-d(1-a)=0$, where $d\equiv \frac{A}{n}$. Taking the derivative with respect to $a$ of both sides, we have
\begin{align*}
    b'(x)x'(a)+\frac{\partial b(x)}{\partial a}+d=0 \Longleftrightarrow  x'(a)=\frac{\frac{\partial b(x)}{\partial a}+d}{-b'(x)}\\
    \Longleftrightarrow  x'(a)=\frac{d+Ax^{\rho}+x\frac{\rho-1}{\rho}(x^{\rho}-1)(ax^{\rho}+1-a)^{\frac{-1}{\rho}}}{-b'(x)}.
\end{align*}
Consider $N(x)\equiv  d+Ax^{\rho}+x\frac{\rho-1}{\rho}(x^{\rho}-1)(ax^{\rho}+1-a)^{\frac{-1}{\rho}}$.
Since $b(x)\equiv aAx^{\rho}+x(ax^{\rho}+1-a)^{1-\frac{1}{\rho}}=d(1-a)$, we have 
\begin{align}\notag
    N(x)\frac{}{}&=\frac{aAx^{\rho}+x(ax^{\rho}+1-a)^{1-\frac{1}{\rho}}}{1-a}+Ax^{\rho}+x\frac{\rho-1}{\rho}(x^{\rho}-1)(ax^{\rho}+1-a)^{\frac{-1}{\rho}}\\
\label{nx3}    &=\frac{Ax^{\rho}}{1-a}+(ax^{\rho}+1-a)^{\frac{-1}{\rho}}x\big( (\frac{1}{1-a}-\frac{1}{\rho})x^{\rho}+\frac{1}{\rho}\big).
\end{align}
Then, 
\begin{align}\label{nx}  
N(x)\frac{(ax^{\rho}+1-a)^{\frac{1}{\rho}}}{x} = \frac{Ax^{\rho-1}}{1-a}(ax^{\rho}+1-a)^{\frac{1}{\rho}}+\big(\frac{1}{1-a}-\frac{1}{\rho}\big)x^{\rho}+\frac{1}{\rho}
\end{align}

Since the lowest fixed point $\bar{k}_1$ is lower than $k^*$, we have $H'(\bar{k}_1)>0$. This implies that $-b'(\bar{k}_1)>0$ (recall that $b'(x)=- H'(x)\frac{A(1-a)}{(H(x))^2}$ for any $x$). By consequence, we have 
$$Sign((\bar{k}_1)'(a))=Sign(N(\bar{k}_1)).$$
Condition $b'(\bar{k}_1)<0$ is equivalent to $ 1-a+\rho a\bar{k}_1^{\rho}+\rho Aa\bar{k}_1^{\rho-1} (a\bar{k}_1^{\rho}+1-a)^{\frac{1}{\rho}}<0$, which implies that 
\begin{align*}
  A\bar{k}_1^{\rho-1} (a\bar{k}_1^{\rho}+1-a)^{\frac{1}{\rho}}>\bar{k}_1^{\rho}+\frac{1-a}{-\rho a}.
\end{align*}
Combining this with (\ref{nx}) and $\rho<0$, we have $N(\bar{k}_1)\frac{(a\bar{k}_1^{\rho}+1-a)^{\frac{1}{\rho}}}{\bar{k}_1}>0.$ So, $(\bar{k}_1)'(a)>0$.	

{\bf The role of $\rho$.}  Consider a fixed point $H(x)=n$, i.e.,  $b(x)\equiv aAx^{\rho}+x(ax^{\rho}+1-a)^{1-\frac{1}{\rho}}=\frac{A}{n}(1-a)$.
Then, we have
\begin{align*}
b'(x)x'(\rho)+\frac{\partial b(x)}{\partial \rho}=0. 
\end{align*}
Thus,
\begin{align*}
    b'(x) x'(\rho) =& -\rho a A x^{\rho - 1}- x (ax^{\rho}+1-a)^{1-\frac{1}{\rho}} \left(\frac{1}{\rho^2} \ln (ax^{\rho } + 1-a) 
    + \frac{ a(\rho-1) x^{\rho-1}}{ax^{\rho}+ 1-a}
    \right). 
\end{align*}
Recall that $b'(\bar k_1)<0$. Hence
\begin{align*}
   Sign(\bar{k}_1'(\rho))=Sign\Big(\rho a A \bar{k}_1^{\rho - 1}+ \bar{k}_1 (a\bar{k}_1^{\rho}+1-a)^{1-\frac{1}{\rho}} \left(\frac{1}{\rho^2} \ln (a\bar{k}_1^{\rho } + 1-a) 
    + \frac{ a(\rho-1) \bar{k}_1^{\rho-1}}{a\bar{k}_1^{\rho}+ 1-a}
    \right)\Big)
\end{align*}
Observe that if $k_1>1$ and $\rho<0$, then $Sign(\bar{k}_1'(\rho))<0$.
    \end{enumerate}
\end{proof}

\subsection{Proofs for Section \ref{intermidiatesection}}
\label{intermidiatesection-proofs}
\begin{proof}[{\bf Proof of Proposition \ref{poverty-middle}}]

 Recall that the dynamics of equilibrium capita path is given by \begin{align}
nk_{t+1}=s(\omega(k_t), f'(k_{t+1}))= \begin{cases}
			s_{\beta_1}(\omega(k_t),  f'(k_{t+1})) & \text{ if }  f'(k_{t+1}) < \frac{1}{\gamma+\beta}\\
			\frac{\omega(k_t)}{1+ f'(k_{t+1})} &\text{ if }\frac{1}{\gamma+\beta} \leq  f'(k_{t+1})\leq \frac{1+\gamma}{\beta}\\
			s_{\beta_2}(\omega(k_t),  f'(k_{t+1}))&\text{ if }  f'(k_{t+1}) >  \frac{1+\gamma}{\beta}.
		\end{cases}
\end{align}


Let $k_0$ satisfy $0<k_0<x_{\beta_1}\leq x_{\beta_2}$ where $x_{\beta_1},x_{\beta_2}$ are defined by Assumption \ref{assum-poverty}.

There are three cases.
\begin{enumerate}
\item If $nk_1=s_{\beta_1}(\omega(k_0),  f'(k_{1}))$ and $  f'(k_{1}) < \frac{1}{\gamma+\beta}$. By Lemma \ref{gbeta},  we have 
\begin{align}
k_1=g_{\beta_1}(k_0).
\end{align}
Since $k_0\in (0,x_{\beta_1})$ (where $x_{\beta_1}$ is defined in Assumption \ref{assum-poverty}), we have $g_{\beta_1}(k_0)<k_0$. So, $k_1<k_0$.

\item If $nk_1=s_{\beta_2}(\omega(k_0),  f'(k_{1}))$ and $  f'(k_{1}) > \frac{1+\gamma}{\beta}$. By Lemma \ref{gbeta},  we have 
\begin{align}
k_1=g_{\beta_2}(k_0).
\end{align}
Since $k_0\in (0,x_{\beta_2})$ (defined in Assumption \ref{assum-poverty}), we have $g_{\beta_2}(k_0)<k_0$. So, $k_1<k_0$.

\item If $nk_{1}=\frac{\omega(k_0)}{1+ f'(k_{1})} $ and $\frac{1}{\gamma+\beta} \leq  f'(k_{1})\leq \frac{1+\gamma}{\beta}$. We have $k_1\geq (f')^{-1}(\frac{1+\gamma}{\beta}).$ By consequence, we have 
\begin{align*}
\omega(k_0)=nk_1\big(1+ f'(k_{1})\big)\geq \left(1+\frac{1}{\gamma+\beta}\right)
\left(f'\right)^{-1}\big(\frac{1+\gamma}{\beta}\big)
\end{align*}
which is a contradiction to our assumption (\ref{omega(k0)}).

\end{enumerate}

To sum up, we obtain that $k_1<k_0$.

By induction, we have $k_{t+1}<k_t<\cdots <k_0$ for any $t\geq 1$. So, $k_t$ converges to some value say $k_*$. Of course, we have $k_*<{x_{poverty}}$. We claim that $k_*=0$. Suppose that $k_*>0$. There are three cases.

1. If $nk_*=\frac{\omega(k_*)}{1+f'(k_*)}$ and $\frac{1}{\gamma+\beta} \leq  f'(k_{*})\leq \frac{1+\gamma}{\beta}$. This cannot happen because of the assumption (\ref{omega(k0)}) and $k_*<k_0$.

2. If $nk_*=s_{\beta_1}(\omega(k_*),f'(k_*))$. This is impossible because $k_*\in (0,x_{\beta_1})$ and Assumption \ref{assum-poverty}.

3. If $nk_*=s_{\beta_2}(\omega(k_*),f'(k_*))$. This is also impossible because $k_*\in (0,x_{\beta_2})$ and Assumption \ref{assum-poverty}.

Finally, we get that $k_*=0$, i.e., the equilibrium capital converges to zero.

{\bf Comparative statics}.	On the one hand, it is known that $x_{\beta_1}$ is decreasing in wariness level $\gamma$. 

In the other hand, we have  $\frac{1+\gamma}{\beta}$ is increasing in $\gamma $ and so are  $\left(f'\right)^{-1} \left(\frac{1+\gamma}{\beta}\right)$ and $L= \left(1+\frac{1}{\gamma+\beta }\right)\left(f'\right)^{-1} \left(\frac{1+\gamma}{\beta}\right) $
because the function $f'$ is decreasing. The increasing property of $\omega$ leads to $\omega^{-1}(L)$ is decreasing in $\gamma$ too. In consequence $x_{poverty} = \min (x_{\beta_1}, \omega^{-1}(L))$ is decreasing in $\gamma$. 
\end{proof}

\begin{proof}[{\bf Proof of Proposition \ref{poverty_x1}}]
		Thanks to Proposition \ref{cesll}, $x_{\beta_1}$ exists if $\rho<0$ and $	-A\rho a^{\frac{1}{\rho}}(1-\rho)^{\frac{1}{\rho}-1} \geq \frac{n(1+\beta_1)}{\beta_1}$. The condition $f'(x_{\beta_1})<\frac{1}{\beta_1}$ guarantees that $x_{\beta_1}$ is a steady state. 

Thanks to Lemma \ref{h_increase} and Proposition \ref{unique-convergence}, there exists a unique equilibrium. Moreover, $k_{t+1}=G(k_t)$, and $k_t$ converges.

According to Corollary \ref{corollary-ces}, we have, for any $t$,
\begin{align*}
k_{t+1}= \begin{cases}
g_{\beta_1}(k_t)& \text{ if } f'(k_{t+1}) < \frac{1}{\gamma+\beta}\equiv \frac{1}{\beta_1}\\
\frac{A(1-a)(ak_t^{\rho}+1-a)^{\frac{1}{\rho}-1}}{n (1+f'(k_{t+1})) } &\text{ if }\frac{1}{\gamma+\beta} \leq f'(k_{t+1})\leq \frac{1+\gamma}{\beta}\\
 			g_{\beta_2}(k_t)&\text{ if } f'(k_{t+1}) >  \frac{1+\gamma}{\beta} \equiv \frac{1}{\beta_2}
 		\end{cases}.
 	\end{align*}
If  (i) $H(k)<k$ $\forall k>0$ or (ii) $\bar k_1>x_{\beta_1}$ where $\bar k_1$ is the smallest positive solution to the equation $H(k)=k$, then  $x_{\beta_1}$ is the smallest fixed value satisfying  $G(x)=x>0$. Then, we immediately have the three statements.
	\end{proof}

\begin{proof}[{\bf Proof of Proposition \ref{steady0}}]
We will prove that $G(k)<k$ for all $k>0$ where $G$ is the dynamic function defined in Proposition \ref{unique-convergence}; this implies that the capital path $(k_t)_{t\geq 0}$ decreasingly converges  to $0$ for all initial capital $k_0$. 
	
We prove a proof for part \ref{part1} of Proposition \ref{steady0}. Similar argument can be applied to the remaining parts. 

Since $\beta_2\leq \beta_1$, we have $\frac{\beta_1}{1+\beta_1}\geq  \frac{\beta_2}{1+\beta_2}$ and then $M_1\geq M_2$. To show that $G(k)<k$ for all $k>0$, we first observe that $s_{\beta_i}(\omega(k), f'(g(k)))< \omega(k)$ for all $k>0$. There are three possible cases.
	\begin{enumerate}
		\item If $f'(G(k))<\frac{1}{\beta_1}$ then we have
		\[
		G(k) = \frac{1}{n}s_{\beta_1}(\omega(k), f'(G(k))) < \frac{\omega(k)}{n} = \left(\frac{\omega(k)}{nk}\right)k \leq \frac{M_1}{n} k <k
		\]
		\item If $f'(G(k))>\frac{1}{\beta_2}$ then using the same argument as above, we have
		\[
		G(k) = \frac{1}{n}s_{\beta_2}(\omega(k), f'(G(k))) < \frac{M_2}{n}k <k
		\]
		\item  If $\frac{1}{\beta_1}\leq f'(G(k))\leq \frac{1}{\beta_2}$ then 
$
			G(k)= \frac{\omega(k)}{n(1+ f'(G(k)))}
$.  Remind that $h(k)= k(1+f'(k))$ is a decreasing function in $k$ and we also have
\begin{align*}
	\frac{h(G(k))}{h(k)} = \frac{\omega(k)}{nh(k)}\leq \frac{M_3}{n}<1.
\end{align*}
Thus $h(G(k))< h(k)$  which implies that $g(k)<k$.

	\end{enumerate}
		By induction, starting at any inital capital $k_0>0$, we have $k_{t+1} <k_t$ for all $t$ and then $(k_t)_{t\geq 0}$ must converge to a limit $k^*$. If $k^*$ is positive then it must satisfy one of the following equations
		\begin{align*}
			nk= s_{\beta_1}(\omega(k), f'(k)) \\
			nk= s_{\beta_2}(\omega(k), f'(k)) \\
			nk = \frac{\omega(k)}{1+ f'(k)}
		\end{align*}
		However under the conditions that $M_i<n$ for $i=1, 2, 3$,  none of these above equations has solution. Therefore the capital path $(k_t)$ must converge to $0$. 
\end{proof}

\section{Logarithmic utility and CES production functions}
\label{specialcase}
As in Corollary \ref{corollary-ces}, assume $u(c)=ln(c)$ and a CES production function: 
\begin{align}	\label{ces}
F(K,L)=A((aK^{\rho}+(1-a)L^{\rho})^{\frac{1}{\rho}}, \text{ where $A>0$, $a\in (0,1)$ and $\rho\not=0,\rho<1$.}\end{align}

 Recall that the elasticity of factor substitution is $\frac{1}{1-\rho}$.\footnote{If $\rho$ approaches 1, we have a linear or perfect substitutes function $(\frac{1}{1-\rho}$ tends to $+\infty)$: $F(K,L)=A((aK+(1-a)L)$.  If $\rho$ approaches zero in the limit, we get the Cobb-Douglas production function:  $F(K,L)=AK^{a}L^{1-a}$.  If $\rho$ approaches negative infinity we get the Leontief or perfect complements production function $(\frac{1}{1-\rho}$ tends to $0$): $F(K,L)=A\min(K,L)$.}

We also present some computations. \begin{align*}
	f(k)&=F(k,1)=A(ak^{\rho}+1-a)^{\frac{1}{\rho}}, \quad 
	f'(k)=Aak^{\rho-1}(ak^{\rho}+1-a)^{\frac{1}{\rho}-1}\\	
		kf'(k)&=Aak^{\rho}(ak^{\rho}+1-a)^{\frac{1}{\rho}-1}\\
			\omega(k)&=f(k)-kf'(k)=
            A(1-a)(ak^{\rho}+1-a)^{\frac{1}{\rho}-1}.
	\end{align*}
We state (without proofs) useful properties of the CES production function.
\begin{lemma}\label{rhopositive}
If $\rho\in (0,1)$, then we have the following properties.
\begin{enumerate}
\item $f(0)=A(1-a)^{\frac{1}{\rho}}$, $f(\infty)=\infty$.  $f'(0)=+\infty$, $f'(\infty)=Aa^{\frac{1}{\rho}}$. 
\item $kf'(k)$ is increasing in $k$. $\lim_{t\to 0}kf'(k)=0$. $\lim_{t\to \infty}kf'(k)=\infty$.
\item $\omega(k)$ is increasing in $k$ because $\rho<1$. $\omega(0)=A(1-a)^{\frac{1}{\rho}}$, $\omega(\infty)=+\infty$

\item Both  $\frac{\omega(k)}{k}$ and $\frac{\omega(k)}{k(1+f'(k))}$ are decreasing in $k$ for $k>0$.
\end{enumerate}

In this case $\rho\in (0,1)$, we can apply results in Proposition \ref{unique-convergence}.
\end{lemma}

\begin{lemma}\label{ces-rho<0}
The case $\rho<0$.

\begin{enumerate}
\item  $f(0^+) =0$, $f(\infty)=A(1-a)^{\frac{1}{\rho}}$. $f'(0)=Aa^{\frac{1}{\rho}}$, $f'(\infty)=0$.
\begin{align}
\lim_{t\to 0}\frac{f(k)}{k}=\lim_{t\to 0}\frac{A(ak^{\rho}+1-a)^{\frac{1}{\rho}}}{k}=\lim_{t\to 0}A(a+\frac{1-a}{k^{\rho}})^{\frac{1}{\rho}}= Aa^{\frac{1}{\rho}}.
\end{align}
\item   $\omega(k)$ is increasing (because $f$ is concave). $\omega(0)=0,$ $\omega(\infty)=A(1-a)^{\frac{1}{\rho}}$ for $\rho<0$.
\item Both $\frac{\omega(k)}{k}$ and $\frac{\omega(k)}{k(1+f'(k))}$ may be -non-monotonic in $k$. 
\end{enumerate}
\end{lemma}


 \begin{proof}[{\bf Proof of Lemma \ref{CesFull}}]
 We denote 
 	\begin{align}
 		W(x)\equiv \frac{\omega(x)}{x}=\frac{A(1-a)(ax^{\rho}+1-a)^{\frac{1}{\rho}-1}}{x} 
 	\end{align}
 	where $A>0,a\in (0,1)$.  	We have \begin{align*}W'(x)=&A(1-a)(ax^{\rho}+1-a)^{\frac{1}{\rho}-2}\frac{1}{x^2}\Big( -\rho ax^{\rho}-(1-a)\Big).
 	\end{align*}
 	
 	\begin{enumerate}
 		\item [(A)] If $\rho>0$ then the function $W$ is strictly decreasing.  Moreover, 
 		\begin{align*}
 			\lim_{x\to 0^+}\frac{\omega(x)}{x}&=\lim_{x\to \infty}\frac{\omega(x)}{x}=\infty,& \quad  \lim_{x\to \infty}\frac{\omega(x)}{x}&=\lim_{x\to \infty}\frac{\omega(x)}{x}=0.
 		\end{align*}
 		By consequence, there exists a unique $x_{\beta}>0$ such that $x_{\beta}=g_{\beta}(x_{\beta})$,  $g_{\beta}(x)>x, \forall x<x_{\beta}$, $g_{\beta}(x)<x, \forall x>x_{\beta}$. (Single crossing property.)
 		
 		\item [(B)] Consider the case $\rho<0$. 
 		
 		Observe that the derivative $W'(x)=0$ for $x=x_0$, $W'(x)>0$ for $x<x_0$, and $W'(x)<0$ for $x>x_0$. So, we have
 		\begin{align*}
 			\max_{x\geq 0}\frac{\omega(x)}{x}&=-A\rho a^{\frac{1}{\rho}}(1-\rho)^{\frac{1}{\rho}-1} \text{ when } x=x_0\equiv \left(\frac{1-a}{-a\rho}\right)^{\frac{1}{\rho}}.
 		\end{align*}
 		Note that $ \lim_{x\to 0^+}\frac{\omega(x)}{x}=\lim_{x\to \infty}\frac{\omega(x)}{x}=0.$
 		
 		Denote $D\equiv D(\beta)\equiv \frac{n(1+\beta)}{\beta}$. 
 		
 		If $\max_{x\geq 0}W(x)<D$, then $\omega(x)<D x$ for any $x>0$.
 		
 		If $\max_{x\geq 0}W(x)=D$, then the equation $\omega(x)=D x$ has a unique  solutions $x_0$, and $\omega(x)<D x$ for any $x\not=x_0,x>0$.

 		If $\max_{x\geq 0}W(x)>D$, then the equation $\omega(x)=D x$ has two solutions $x_1, x_2$ with for $0<x_1<x_0<x_2$. Moreover,  $\omega(x)<D x$ for $x\in (0,x_1)\cup (x_2,\infty)$ and  $\omega(x)>D x$ for $x\in (x_1,x_2)$. 
 		
 	\end{enumerate}

 	We see that $W'(x)=0$ iff $x=x_0$. $W'(x)>0$ iff $x<x_0$. So, $\max_{x\geq 0}W(x)=W(x_0)=-A\rho a^{\frac{1}{\rho}}(1-\rho)^{\frac{1}{\rho}-1}$.

{\bf Comparative statics}. $x_1$ is decreasing in $\beta$ but increasing in $a$. By Lemma \ref{poverty-function-g}, $x_1$ is decreasing in $\beta$.

We have $0<x_1<x_0<x_2$, $x_i=g_{\beta_1}(x_i)$ for $i=1,2$. Consider a fixed point $x$
\begin{align}\label{fixedpointx}
x=g_{\beta}(x)= \frac{\beta A(1-a)(ax^{\rho}+1-a)^{\frac{1}{\rho}-1}}{n(1+\beta)}.
\end{align}
Denote $d\equiv \frac{n(1+\beta)}{A\beta}$. We have $(1-a)(ax^{\rho}+1-a)^{\frac{1}{\rho}-1}-dx=0$. Taking the derivative with respect to $a$ of both sides of this equation, we have
\begin{align*}
  -dx'(a)-  (ax^{\rho}+1-a)^{\frac{1}{\rho}-1}+(1-a)(\frac{1}{\rho}-1) \big(x^{\rho}-1+a\rho x^{\rho-1}x'(a)\big)(ax^{\rho}+1-a)^{\frac{1}{\rho}-2}&=0\\
\Leftrightarrow  x'(a) =\frac{
  (ax^{\rho}+1-a)^{\frac{1}{\rho}-1} -(1-a)(\frac{1}{\rho}-1)\big(x^{\rho}-1\big)(ax^{\rho}+1-a)^{\frac{1}{\rho}-2}}{\Big(-d+(1-a)(1-\rho)ax^{\rho-1}(ax^{\rho}+1-a)^{\frac{1}{\rho}-2}\Big)}.
\end{align*}
Recall that for $x_1$, we have $(1-a)(ax_1^{\rho}+1-a)^{\frac{1}{\rho}-1}-dx_1=0$ and $x_1<x_0\equiv \left(\frac{1-a}{-a\rho}\right)^{\frac{1}{\rho}}$.
Denote $X\equiv ax_1^{\rho}+1-a$. Look at the denominator in the formula of $x'(a)$, we have
\begin{align*}
&-d+(1-a)(1-\rho)ax_1^{\rho-1}(ax^{\rho}+1-a)^{\frac{1}{\rho}-2}=\frac{1-a}{x_1}X^{\frac{1}{\rho}-2}(a(1-\rho)x_1^{\rho}-X)\\
&=\frac{1-a}{x_1}X^{\frac{1}{\rho}-2}((-a\rho)x_1^{\rho}-(1-a))>0
\end{align*}
because $x_1<x_0\equiv \left(\frac{1-a}{-a\rho}\right)^{\frac{1}{\rho}}$ and $\rho<0$, which imply that $x_1^{\rho}>\frac{1-a}{-a\rho}$. 
Now, we look at the numerator.
\begin{align*}
  N&\equiv (ax_1^{\rho}+1-a)^{\frac{1}{\rho}-1} -(1-a)(\frac{1}{\rho}-1)\big(x_1^{\rho}-1\big)(ax_1^{\rho}+1-a)^{\frac{1}{\rho}-2}\\
  &=-X^{\frac{1}{\rho}-2}\big((1-a)(1-\rho)\frac{1}{\rho}(x_1^{\rho}-1)-(ax_1^{\rho}+1-a)\big)\\
  &= -X^{\frac{1}{\rho}-2}\frac{(1-a-\rho)x_1^{\rho}-(1-a)}{\rho}.
\end{align*}
Observe that $\frac{1-a}{1-a-\rho}<\frac{1-a}{-a\rho}$. Combining with $x_1<x_0\equiv \left(\frac{1-a}{-a\rho}\right)^{\frac{1}{\rho}}$, we get that $N>0$. By consequence, we have $x_1'(a)>0$.


{\bf  Role of $\rho$}. We have $(1-a)(ax^{\rho}+1-a)^{\frac{1}{\rho}-1}-dx=0$, or, equivalently, $\ln(\frac{d}{1-a})+\ln(x)=(\frac{1}{\rho}-1)\ln(ax^{\rho}+1-a)$. Taking the derivative with respect to $\rho$ of both sides,   we have
\begin{align*}
\frac{x'(\rho)}{x}=\frac{-1}{\rho^2}\ln(ax^{\rho}+1-a)+\frac{\frac{1}{\rho}-1}{ax^{\rho}+1-a}a\frac{\partial x^\rho}{\partial x}.
\end{align*}
Recall that $\frac{\partial x^\rho}{\partial x}=x^{\rho}(\ln(x)+\frac{\rho}{x}x'(\rho))$. Thus, we have
\begin{align*}
    x'(\rho)\Big(1-\frac{\frac{1}{\rho}-1}{ax^{\rho}+1-a}\rho ax^{\rho}\Big)=\frac{-1}{\rho^2}\ln(ax^{\rho}+1-a)+(\frac{1}{\rho}-1)\frac{ax^{\rho}\ln(x)}{ax^{\rho}+1-a}.
\end{align*}

Focus on the minimum fixed point $x_1$. We have $<_1<x_0\equiv \left(\frac{1-a}{-a\rho}\right)^{\frac{1}{\rho}}$. By consequence, $1-\frac{\frac{1}{\rho}-1}{ax_1^{\rho}+1-a}\rho ax_1^{\rho}<0$. Therefore, \begin{align}\label{sign1}
    Sign(x_1'(\rho))&=-Sign\Big(\frac{-1}{\rho^2}\ln(ax_1^{\rho}+1-a)+(\frac{1}{\rho}-1)\frac{ax_1^{\rho}\ln(x_1)}{ax_1^{\rho}+1-a}\Big).
\\
\label{sign2}
&=Sign\Big((ax_1^{\rho}+1-a)\ln(ax_1^{\rho}+1-a)-(1-\rho)ax_1^{\rho}\ln(x_1^{\rho})\Big).
\end{align}
Denote $y_1\equiv x_1^{\rho}$ and define the function $B(y)\equiv (ay+1-a)\ln(ay+1-a)-(1-\rho)ay\ln(y)$.

We have 
\begin{align}
B'(y)&=\rho a+a\big(\ln(ay+1-a)-(1-\rho)\ln(y)\big)\\
B^{\prime\prime}(y)&=\frac{a}{ay+(1-a)y}\big(\rho a y-(1-\rho)(1-a)\big).
\end{align}
Since $\rho<0$, we have $B^{\prime\prime}(y)<0$. 
Observe that 
\begin{align}
    B(0)&=(1-a)\ln(1-a)<0,& B(1)&=0, & B(\infty)&=-\infty,\\
    B'(0)&=\infty,& B'(1)&=\rho a<0,&B'(\infty)&=-\infty.
\end{align}
 So, there exists a unique $y_*\in (0,\infty)$ such that $B'(y_*)=0$, $B'(y)>0$ if $y<y_*$, $B'(y)<0$ if $y>y_*$.  Note that $y_*\in (0,1)$ (because $B'(0)=\infty, B'(1)=\rho a<0$) and it only depends on $\rho,a$.  Then, $B(y_*)$ is the maximum value of $B(y)$.  Recall that $B(1)=0$. Of course, $B(y_*)>0$. 
 
 Therefore, there exists a unique $y_s\not =1$ such that $B(y_s)=B(1)=0$. Moreover, $y_s\in (0,y_*).$ $B(y)>0$ iff $y\in (y_s,1)$. $B(y)<0$ iff $y\in (0,y_s)\cup (1,\infty)$.  By consequence, we obtain claims (\ref{cesrho1}) and (\ref{cesrho2}) in Lemma \ref{CesFull}.

We now prove a last statement in claim (\ref{cesrho1}) in Lemma \ref{CesFull}. Observe that if $x_1<1$, then $\frac{-1}{\rho^2}\ln(ax_1^{\rho}+1-a)+(\frac{1}{\rho}-1)\frac{ax_1^{\rho}\ln(x_1)}{ax_1^{\rho}+1-a}>0$, and hence $x_1'(\rho)<0$.

We now assume that  $x_0<1$,  which happens iff $\rho>-\frac{a}{1-a}$, i.e., $\rho$ is not so small (or, equivalently, the elasticity of factor  substitution $\frac{1}{1-\rho}$ is quite high, in the sense that $\frac{1}{1-\rho}>1-a$), then we have $x_1<1$ (because $x_1<x_0$). Consequently, $x_1'(\rho)<0$ if $x_0<1$.

\end{proof}

\end{document}